\documentclass[12pt]{article}

\usepackage{amsfonts, amsmath, amssymb}
\usepackage{stmaryrd}
\usepackage{a4wide,color}
\usepackage{tikz}
\usetikzlibrary{positioning}
\usepackage{bm}
\usepackage{hyperref}
\usepackage{float}
\usepackage{comment}

\usepackage{multicol}

\newcommand{\eq}{\leftrightarrow}
\newcommand{\Eq}{\Leftrightarrow}
\newcommand{\imp}{\rightarrow}
\newcommand{\Imp}{\Rightarrow}

\newcommand{\vel}{\vee}

\newcommand{\satisfies}{\vDash}

\newcommand{\T}{\top}
\newcommand{\F}{\bot}

\newcommand{\dia}[1]{\langle #1 \rangle}
\renewcommand{\phi}{\varphi}
\newcommand{\union}{\cup}

\newcommand{\ModelM}{M=(W,\sim,V)}

\newcommand{\Ka}{\hat{K}_a}
\newcommand{\M}{\hat{K}}

\newcommand{\view}{\vartriangleright_a}
\newcommand{\Hist}{\mathcal{H}}
\newcommand{\Words}{\mathcal{W}}
\newcommand{\proj}{{\upharpoonright}}
\newcommand{\tri}{\triangleright}
\newcommand{\prefix}{\sqsubseteq}

\newcommand{\emp}{\mathsf{empty}}

\newcommand{\AAstar}{\textbf{AA}^\ast}

\newcommand{\Nat}{\mathbb N}

\newcommand{\p}{p}
\newcommand{\q}{q}
\newcommand{\np}{\overline{p}}

\newcommand{\nq}{\overline{q}}

\usepackage[thmmarks]{ntheorem}
\theoremsymbol{\ensuremath{\dashv}}
\usepackage{newproof}

\newtheorem{theorem}{Theorem}
\newtheorem{lemma}[theorem]{Lemma}
\newtheorem{proposition}[theorem]{Proposition}

\newtheorem{corollary}[theorem]{Corollary}

\newtheorem{example}[theorem]{Example}
\newtheorem{definition}[theorem]{Definition}

\newcommand{\lang}{\mathcal L}

\newcommand{\A}{\mathcal{A}}

\newcommand{\weg}[1]{}

\usepackage[backend=biber, style=trad-plain, sorting=nyt]{biblatex}
\begin{document}

\title{Axiomatisation for an asynchronous epistemic logic with sending and receiving messages}
\author{Philippe Balbiani \and Hans van Ditmarsch \and Clara Lerouvillois\footnote{Affiliations of P. Balbiani and H. van Ditmarsch: IRIT, CNRS---INPT---UT, Toulouse, France. Of C. Lerouvillois: IRIT, CNRS---INPT---UT, Toulouse, France \& IHPST University Paris 1 Panthéon Sorbonne, Paris, France. Clara Lerouvillois is corresponding author, email: clara.lerouvillois@irit.fr}}
\date{}

\maketitle

\begin{abstract}
We investigate a logic for asynchronous announcements wherein the sending of the messages by the environment is separated from their reception by the individual agents. Both come with different modalities. In the logical semantics, formulas are interpreted in a world of a Kripke model but given a history of prior announcements and receptions that already happened. An axiomatisation \textbf{AA} for such a logic has been given in prior work, for the formulas that are valid when interpreted in the Kripke model before any such announcements have taken place.  This axiomatisation is a reduction system wherein one can show that every formula is equivalent to a purely epistemic formula without dynamic modalities for announcements and receptions. We propose a generalisation $\AAstar$ of this axiomatisation, for the formulas that are valid when interpreted in the Kripke model given any history of prior announcements and receptions of announcements. It does not extend the axiomatisation \textbf{AA}, for example it is no longer valid that nobody has received any message. Unlike \textbf{AA}, this axiomatisation $\AAstar$ is infinitary and it is not a reduction system. \\
%
%
%

\end{abstract}

\section{Introduction}

What does an agent know in a dynamic setting and how does her knowledge evolve through  communication in the absence of a global clock?
Dynamic epistemic logics (DEL) are modal logics of knowledge and change of knowledge. Some studies enforce \emph{synchrony} \cite{Benthem2007merging} for such logics, whereas others accommodate \emph{asynchrony} \cite{degremont2011synchronicity}. 

There are different ways to accommodate asynchrony in epistemic logics. 

Given epistemic actions that are not public, such as private announcements, one way to model asynchrony is that action sequences of different length are indistinguishable for an agent. If we identify action execution with a clock tick, this then represents uncertainty over the time. Such asynchrony is found, for example, in the gossip protocols of \cite{apt2016epistemic} (although not a DEL), wherein agents exchange sets of secrets in peer-to-peer communications (telephone calls) of which other agents may be unaware; and in \cite{Ditmarsch2010Prisoners} modelling the {\em One Hundred Prisoners} epistemic puzzle, wherein agents flip a light switch during an interrogation while other agents remain uncertain about the number of interrogations (if any) that have already taken place.

A different way to enforce asynchrony, more akin to assumptions in distributed computing \cite{lamport1978ordering,kshemkalyani2011distributed}, is to consider sending and receiving messages as separate actions. In DEL this is typically not the case: the epistemic action there, such as the public announcement in {\em public announcement logic} PAL \cite{plazaPAL}, should be seen as instantaneous reception by some or all agents of messages sent by the environment. Such DEL are logics of observation, not of messaging (nor of agency). However, recent work in DEL have proposed logics containing different modalities for sending and receiving messages \cite{knight2019reasoning,AA}. Our work builds on their efforts and results. In \cite{knight2019reasoning,AA}, messages are publicly broadcast and individually received by the agents. Other works also allow partial synchronization wherein a subset of the set of all agents simultaneously receive a sent message, thus bridging the gap between asynchrony and synchrony \cite{fromPALtoAA}.

As an example, let us say a new podcast series has premiered on a topic that interests three friends Alice, Bob and Charlie. Each episode of this podcast---the message, so to speak---is released at irregular intervals on a podcast hosting platform, and these are meant to be listened to in the order in which they were broadcast. If Alice has listened to the first podcast episode by herself, she is uncertain whether Bob and Charlie have also listened to it. In fact, she can imagine different \emph{histories}: one in which Bob and Charlie have listened to the episode before her, or only Bob, or after her, and so on.

We propose structures in which the agents not only may have epistemic uncertainty over different worlds, but also temporal uncertainty over such worlds, which is represented by a different, orthogonal, binary relation. The second kind of uncertainty is used to reason over different histories of epistemic actions (of possibly different length). 

We make a number of further assumptions in our knowledge representation wherein we follow the approach in \cite{AA}.
First, as said, agents receive messages in the order in which they were sent. In this we follow the classic FIFO scheme of communication in asynchronous systems \cite{brand1983communicating}.
Second, when envisaging alternative histories of past actions, agents only consider the messages they have already received. For example, Alice cannot assume that Bob has listened to the third episode of the podcast if she has not listened to it yet. We assume that she has no knowledge that any further episodes will ever be released.
This  assumption fits well our setting wherein sending and receiving messages that are announcements reduce agents' uncertainty about unchanging facts---contrary to communication in distributed systems wherein messages may also change the value of facts and with that agents' uncertainty over such factual change.
Notice from this assumption we can get interesting mathematical properties, namely a neat semantics, avoiding the circularity problem encountered in \cite{knight2019reasoning}, and a complete axiomatisation. 
See also Section \ref{sectionComparison} for a detailed discussion.
Third, we assume that announcements are \emph{truthful}, meaning they are true when broadcast, as in PAL. Related to that, we need a notion of \emph{executability} of histories. This notion ensures that agents only consider possible histories that consist of announcements that, given a state of the system, can be truthfully broadcast and received there.

Intuitively, in our approach, an agent knows a formula if and only if she can only imagine states and histories executable there that satisfy the formula, taking into account that the messages she received were true when sent, while ignoring that other agents may have received more announcements than herself.

In \cite{AA}, and the related \cite{fromPALtoAA,quantifyingAA}, two notions of validity are defined: a formula is {\em $\epsilon$-valid} (\emph{valid}) if it is true in every state of every epistemic model (Kripke model), given that no sending or receiving actions have yet been executed. In other words, these formulas are always true given the empty history $\epsilon$. This carries no intuition of temporal uncertainty, but assumes a commonly known origin to start the interpretation. Then, a formula is \emph{$\ast$-valid} (\emph{always-valid}) if it is true in every state of every epistemic model, and given any prior history of actions (sending and receiving). The $\epsilon$-validities have been axiomatised in \cite{AA}.

The present work extends \cite{AA} by proposing an axiomatisation {\bf AA}$^\ast$ for always-validities ($\ast$-validities). We also slightly modify their semantics for asynchronous announcements so that formulas can only be true in (state, history) pairs such that the history is executable in that state. This does not modify the set of validities and therefore does not affect the axiomatisation. Unlike the axiomatisation {\bf AA} of \cite{AA}, that is a rewrite system reducing every formula with dynamic modalities for sending and receiving to one without, our novel axiomatisation {\bf AA}$^\ast$ is an infinitary axiomatisation from which dynamic modalities cannot be eliminated. An extensive final section compares our results to other works in the area.


Section \ref{sectionPresentationAA} presents the logic of \emph{asynchronous announcements}. Section \ref{sectionAxiomatisation} proposes an axiomatisation $\AAstar$ for always-validities, which completeness is shown in Section \ref{sectionCompleteness}. Section \ref{sectionAxiomOneAgent} discusses why we do not have an axiomatisation for the single-agent case and Section \ref{sectionComparison} compares our work to other approaches on asynchronous communication and three-valued logics.


\section{The logic of asynchronous announcements} \label{sectionPresentationAA}

In this section we present the language and the semantics for asynchronous announcements. The presentation is based on \cite{AA}, except for the definition of the satisfiability and executability relations (see Definition \ref{defSemantics}).

\subsection{Syntax} \label{sectionSyntax}

We first define the language of asynchronous announcements and then the notions of \emph{word} and \emph{history} that we use to represent sequences of sending and receiving events.

\begin{definition}[Language $\lang$]
Let $P$ be a countable set of atoms (denoted p, q, etc.) and $\A$ be a finite set of agents (denoted a, b, etc.). The language $\lang$ of asynchronous announcement logic is defined as follows.
\begin{displaymath}
\phi ::= p\ |\ \T \ |\ \lnot\phi\ |\ (\phi \lor \phi) \ |\ \Ka \phi\ |\ \langle \phi \rangle \phi\ |\ \langle a \rangle \phi\ 
\end{displaymath}
\end{definition}

We follow the standard rules for omission of the parentheses. Intuitively, $\langle \phi \rangle \psi$ means that $\phi$ is announced and, after that, $\psi$ holds. Similarly, $\langle a \rangle \phi$ means that after agent $a$ effectively receives a new message (the `next one in the queue'), $\phi$ holds. Without modalities $\langle a \rangle$ we get the language $\lang_{PAL}$ of public announcement logic; and without $\langle \phi \rangle$ we get the language $\lang_{ml}$ of multi-agent modal logic.

As usual, we define $\F := \lnot \T$, $\phi \land \psi := \lnot (\lnot \phi \lor \lnot \psi)$, $\phi \imp \psi := \lnot \phi \lor \psi$ and $\phi \eq \psi := (\phi \imp \psi) \land (\psi \imp \phi)$. The dual of $\Ka$ is defined by abbreviation as $K_a \phi := \lnot \Ka \lnot \phi$.
We also define by abbreviation the dual of dynamic modalities as $[\phi]\psi := \lnot \langle \phi \rangle \lnot \psi$ and $[a]\phi := \lnot \langle a \rangle \lnot \phi$.

Since we want to model \emph{asynchronous} communicative situations, we need to define further notions to distinguish different possible orders of message reception by the agents: the notions of \emph{word} and \emph{history}.\\

Consider $\A \cup \lang$ as an alphabet with agents and formulas as letters. Words $\alpha, \beta,...$ over $\A \cup \lang$ are finite sequences of symbols over $\A \cup \lang$. The empty word is denoted $\epsilon$.
Let $\Words := (\A \cup \lang)^\ast$ be the set of all words.

For clarity, we add dots to separate letters within a word, \emph{e.g.} $p.\lnot K_a p.a.q.a$. When there is no ambiguity, however, we omit the point, particularly when abbreviations are used, as in $\alpha a$ or $\alpha \phi$.

Given a word $\alpha$, we use the following notations for intuitive notions that can be easily defined by induction: $|\alpha|$ is its length; $|\alpha|_a$ is the number of occurrences of $a$ in $\alpha$; $|\alpha|_!$ is the number of its formula occurrences and $|\alpha|_{!a}$ the number of formula occurrences received by agent $a$. In the single-agent case, \emph{i.e.} when $\A = \lbrace a \rbrace$, it is clear that $|\alpha| = |\alpha|_! + |\alpha|_a$. Otherwise, in the multi-agent case, \emph{i.e.} when $\A = \lbrace a_1, \cdots, a_n \rbrace$, $|\alpha| = |\alpha|_{a_1} + \cdots + |\alpha|_{a_n} + |\alpha|_!$. However, in general $|\alpha|_{!a}\neq |\alpha|_a$: for example, if $\alpha=a.p.q.a.r.a$, then $|\alpha|_a=3$ but $|\alpha|_{!a}=2$.
We further define $\alpha \proj_!$ as the projection of $\alpha$ to $\lang$---hence $\alpha\proj_!$ is the word obtained from $\alpha$ by retaining occurrences of formulas only--- and $\alpha \proj_{!a}$ as the restriction of $\alpha \proj_!$ to the first $|\alpha|_{!a}$ occurrences of formulas so $\alpha \proj_{!a}$ is the restriction of $\alpha \proj_{!}$ to the formulas that agent $a$ has read. Consider for instance the word $\alpha=p.q.a.\lnot K_b q.b.a$. Then $\alpha\proj_! = p.q.\lnot K_b q$ and $\alpha\proj_{!a}=p.q$ whereas $\alpha\proj_{!b}=p$.
Finally, given a word $\alpha$ and $n \in \mathbb{N}$, $\alpha^n$ denotes a concatenation of $n$ copies of $\alpha$. For example, if $\alpha = p.q.a$ then $\alpha^2 = p.q.a.p.q.a$.

Note that a non-empty word $\alpha$ can be decomposed into $\alpha'  \mu$ or $\mu \alpha'$ for $\mu$ a symbol in $\A \cup \lang$. In future proofs, if proceeding by induction on a word $\alpha$, we will use one or the other decomposition.

For all words $\alpha$ over $\A \union \lang$, the modality $\langle \alpha \rangle$ is inductively defined by $\langle \epsilon \rangle \phi  :=  \phi$,  $\langle \alpha  a \rangle \phi  :=  \langle \alpha \rangle \langle a \rangle \phi$ and $\langle \alpha \psi \rangle  :=  \langle \alpha \rangle \langle \psi \rangle \phi$. Its dual is defined by abbreviation as $ [\alpha]\phi := \lnot \langle \alpha \rangle \lnot \phi$.

\begin{definition}[Prefix]
A word $\beta$ is a \emph{prefix} of a word $\alpha$, denoted $\alpha \prefix \beta$, if $\beta$ is an initial sequence of $\alpha$. Obviously, $\alpha \prefix \alpha$ and, if $\beta \prefix \alpha$, then for all $a \in \A$ and $\phi \in \lang$, $\beta \prefix \alpha a$ and $\beta \prefix \alpha\phi$.
\end{definition}

We assume that agents read announcements in the order in which they were sent. Words wherein that is the case are called \emph{histories}. Therefore, a history is a word such that, for each agent, any prefix contains more sent messages than reception modalities. Formally:

\begin{definition}[History]\label{defHistory}
A word $\alpha$ over $\A \cup \lang$ is a history if and only if $|\beta|_a \leq |\beta|_!$ for all agents $a \in \A$ and for all prefixes $\beta \prefix \alpha$.
We call $\Hist$ the set of histories over $\A \cup \lang$.
\end{definition}
Therefore, $p. q. a. a$ and $p. a. q$ are histories, but $p. a. a. q$ and $p. q. a. a. a$ are \emph{not}.
Obviously, if $\alpha$ is a history, then $|\alpha|_{!a}= |\alpha|_a$ so $\alpha \proj_{!a}$ is simply the restriction of $\alpha \proj_!$ to the first $|\alpha|_a$ formula occurrences.

\begin{lemma} \label{lemmaPrefixHistory}
Let $\alpha$ be a word over $\A \cup \lang$ and $\beta \prefix \alpha$ a prefix of $\alpha$. If $\alpha$ is a history, then $\beta$ is also a history.
\end{lemma}

\begin{proof}
Let $\alpha, \beta$ be two words over $\A \cup \lang$ such that $\alpha \in \Hist$ and $\beta \prefix \alpha$. Let $\gamma \prefix \beta$ be a prefix of $\beta$. Then $\gamma$ is also a prefix of $\alpha$. Hence, by Definition \ref{defHistory}, $|\gamma|_a \leq |\gamma|_!$ for all agents $a \in \A$. Since this holds for any arbitrary prefix of $\beta$, $\beta$ is a history.
\end{proof}

To define an appropriate semantics for knowledge through asynchronous announcements, we further need the following \emph{view relation} between words and histories.

\begin{definition}[View relation]\label{defViewRelation}
    For every agent $a\in \A$, the view relation $\view$ is defined on $\Words \times \Hist$ as follows: $\alpha \view \beta$ if and only if $\beta\proj_! = \beta\proj_{!a} = \alpha\proj_{!a}$.
    For any word $\alpha$ and agent $a$ we define the set of histories \textbf{view}$_a(\alpha):=\lbrace \beta\in \Hist \ |\ \alpha \view \beta \rbrace$.
\end{definition}


Roughly, for a given actual history $\alpha$, \textbf{view$_a$}$(\alpha)$ is the set of histories that agent $a$ considers possible, in which all the messages that have been broadcast are precisely those agent $a$ has currently received. It will then be natural to define knowledge depending on such uncertainty over histories (see Section \ref{sectionSemantics}).

\begin{example}
Let us consider only two agents: $\A = \lbrace a, b \rbrace$. If the actual history is $\alpha = p.a$ then agent $a$ may imagine that agent $b$ has also received the announcement (either before or after $a$). Hence, \textbf{view}$_a(\alpha) = \lbrace p.a, p.a.b, p.b.a \rbrace$. However, agent $b$ has no idea that a message has been sent, so \textbf{view}$_b(\alpha) = \lbrace \epsilon \rbrace$.
\end{example}

We end this section by stating some interesting properties of the view relation. We recall that a binary relation $R$ on a given set $X$ is Euclidean if and only if for all $x,y,z \in X$, whenever $xRy$ and $xRz$, $yRz$ also holds; and $R$ is post-reflexive if and only if for all $x,y \in X$, if $xRy$ then $yRy$.

\begin{proposition}
The view relation is serial, transitive, Euclidean and post-reflexive.
Moreover, for all words $\alpha$, the set $\textbf{view}(\alpha)$ is finite.
\end{proposition}

\noindent Note, however, that the view relation is neither reflexive nor symmetric. Indeed, $p. a. q \not\view p. a. q$ and $p. a. q \view p. a$ but $p. a \not\view p. a. q$.

\subsection{Structures}

In this section, we present the structures on which we interpret the formulas.

\begin{definition}[Epistemic model]
An epistemic model is a triple $\ModelM$ where
\begin{itemize}
\item $W \neq \emptyset$ is a set of states
\item $\sim: \A \longrightarrow \mathcal{P}(W^2)$ assigns to each agent $a \in \A$ an accessibility relation $\sim_a$ on W
\item $V: P \longrightarrow \mathcal{P}(W)$ is a valuation, assigning to each atom $p \in P$ a set $V(p) \subseteq W$
\end{itemize}
\end{definition}

Although it is possible to work with arbitrary accessibility relation, here we only consider equivalence relations, \emph{i.e.} for all agents $a$, $\sim_a$ is reflexive, transitive and symmetric, or, equivalently, reflexive and Euclidean.

\begin{example}\label{example1} \hfill
\begin{multicols}{2}
\noindent This figure models a system where two agents, Alice and Bob, are only aware of their local variable, respectively $p$ and $q$. Instead of naming the states we show their valuation of $p$ and $q$. We write $\np$ for $\lnot p$ and omit reflexive arrows for clarity.
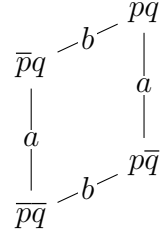
\begin{center}
\begin{tikzpicture}
\node (00) at (0,0) {$\np\nq$};
\node (01) at (0,2) {$\np\q$};
\node (10) at (1.5,.7) {$\p\nq$};
\node (11) at (1.5,2.7) {$\p\q$};
\draw (00) -- node[fill=white,inner sep=1pt] {$b$} (10);
\draw (01) -- node[fill=white,inner sep=1pt] {$b$} (11);
\draw (00) -- node[fill=white,inner sep=1pt] {$a$} (01);
\draw (10) -- node[fill=white,inner sep=1pt] {$a$} (11);
\end{tikzpicture}
\end{center}
\end{multicols}
\end{example}

\subsection{Semantics}\label{sectionSemantics}

We now present the semantics for asynchronous announcements, define two notions of validities and state some important properties of the semantics.

To define the semantics, we need a well-founded order $\ll$ between pairs $(\alpha,\phi)$ of word and formula. This order uses two auxiliary functions: $\Vert \cdot \Vert$ represents the \emph{size} of formulas or words, and $deg(\cdot)$ displays the \emph{modal depth} of formulas. Both are defined in the following.

For all $\phi \in \lang$, $\Vert \phi \Vert$ is inductively defined by
\begin{align*}
\Vert p \Vert &:= 2 & \Vert \phi \lor \psi \Vert &:= \Vert \phi \Vert + \Vert \psi \Vert 
	& \Vert \langle a \rangle \phi \Vert &:= \Vert \phi \Vert +2 \\
\Vert \T \Vert &:= 1 & \Vert \Ka \phi \Vert &:= \Vert \phi \Vert + 1
	& \Vert \langle \psi \rangle \phi\Vert &:= 2\Vert \psi \Vert + \Vert \phi \Vert \\
\Vert \lnot \phi \Vert &:= \Vert \phi \Vert + 1 & &
\end{align*}

\noindent and for all words $\alpha$ over $\A \union \lang$, $\Vert \alpha \Vert$ is defined by
\begin{align*}
\Vert \epsilon \Vert &:= 0 & \Vert \alpha \Vert &:= \sum_{a\in \A} |\alpha|_a + \sum_{\phi \in \alpha\proj_!} \Vert\phi\Vert
\end{align*}

\noindent Then, for all formulas $\phi \in \lang$, $deg(\phi)$ is inductively defined as follows:
\begin{align*}
deg(p) &:= 0
	& deg(\Ka \phi) &:= deg(\phi) + 1\\
deg(\T) &:= 0
	& deg(\langle a \rangle \phi) &:= deg(\phi)  \\
deg(\lnot \phi) &:= deg(\phi) 
	& deg(\langle \psi \rangle \phi) &:= deg(\psi) + deg(\phi) \\
deg(\phi \lor \psi) &:= max( deg(\phi),deg(\psi))
&
\end{align*}

\noindent Also, given a pair $(\alpha, \phi) \in \Words\times \lang$, we define:
\begin{displaymath}
deg(\alpha, \phi) := deg(\langle \alpha \rangle \phi)
\end{displaymath}

\noindent Finally, the well-founded order $\ll$ is defined between pairs $(\alpha,\phi)\in \Words\times \lang$:
\begin{align*}
(\alpha, \phi) \ll (\beta,\psi) \text{ iff }  &deg(\alpha, \phi) < deg(\beta, \psi)\\
\text{ or } &deg(\alpha, \phi) = deg(\beta, \psi) \text{ and } \Vert \alpha \Vert + \Vert \phi \Vert < \Vert \beta \Vert + \Vert \psi \Vert\\
\end{align*}

Useful results about this order can be found in \cite{AA}.
Here, we only show the following property:

\begin{lemma}\label{lemmaOrder}
Let $\alpha$ be a word and $\phi \in \lang$ be a formula. Then:
\begin{displaymath}
deg(\langle \alpha \rangle \phi) = deg(\phi) + \sum_{\psi \in \alpha} deg(\psi)
\end{displaymath}
\end{lemma}

\begin{proof} We show it by induction on $\Vert \alpha \Vert$. Let $\alpha \in \Words$ be a word such that for all $\beta \in \Words$, if $\Vert \beta \Vert < \Vert \alpha \Vert$, then $deg(\langle \beta \rangle \phi) = deg(\phi) + \sum_{\psi \in \beta} deg(\psi)$.
We now show that this property also holds for $\alpha$. We distinguish three cases:
\begin{itemize}
\item Case $\alpha = \epsilon$. Obviously, for all formulas $\phi \in \lang$, $deg(\langle \epsilon \rangle \phi) = deg(\phi)$.
\item Case $a \alpha$. We have the following:
\begin{align*}
deg(\langle a \alpha \rangle \phi)  &= deg(\langle a \rangle \langle \alpha \rangle \phi) \\
									&= deg(\langle \alpha \rangle \phi) \hspace{6.4em} \text{by definition} \\
                               		&= deg(\phi) +  \sum_{\psi \in \alpha} deg(\psi) \qquad \text{by induction hypothesis, because } \Vert\alpha\Vert < \Vert a\alpha \Vert \\
                               	  	&= deg(\phi) + \sum_{\psi \in a \alpha} deg(\psi) \qquad \text{because } (\alpha a)\proj_! = \alpha\proj_!\\
\end{align*}
\item Case $\chi \alpha$. As above, we have:
\begin{align*}
deg(\langle \chi \alpha \rangle \phi) &= deg(\langle \chi \rangle \langle \alpha \rangle \phi)\\
								 &= deg(\chi) + deg(\langle \alpha \rangle \phi) \hspace{4.7em} \text{by definition} \\
                                 &= deg(\chi) + deg(\phi) + \sum_{\psi \in \alpha} deg(\psi) \ \text{by induction hypothesis, because } \Vert\alpha \Vert < \Vert \chi\alpha \Vert\\
                                 &= deg(\phi) + deg(\chi) + \sum_{\psi \in \alpha} deg(\psi) \\
                                 &= deg(\phi) + \sum_{\psi \in \chi \alpha} deg(\psi)
\end{align*} 
\end{itemize}
\end{proof}

We can now define the semantics for asynchronous announcements.
To do so, we also introduce an \emph{executability relation} $\bowtie$ that is used to express the fact that a given history can indeed be executed in a given state.

\begin{definition}[Semantics]\label{defSemantics}
Let $\ModelM$ be an epistemic model. We simultaneously define the relation $\bowtie$ between states $s \in W $ and words $\alpha \in \Words$ and the relation $\satisfies$ between pairs $(s,\alpha)$ of states and words, and formulas $\phi \in \lang$ by $\ll$-induction:

\begin{align*}
&s \bowtie \epsilon      & &\text{always} \\
&s \bowtie \alpha a      & &\text{iff} & &s \bowtie \alpha \text{ and  } |\alpha|_a < |\alpha|_!\\
&s \bowtie \alpha\phi    & &\text{iff} & &s\bowtie \alpha \text{ and } s,\alpha\satisfies\phi  \\ \\
&s,\alpha \satisfies p                            & &\text{iff} & &s \bowtie \alpha \text{ and } s \in V(p) \\
&s,\alpha \satisfies \T                           & &\text{iff} & &s \bowtie \alpha \\
&s,\alpha \satisfies \lnot \phi                   & &\text{iff} & &s\bowtie \alpha \text{ and } s,\alpha \nvDash \phi \\
&s,\alpha \satisfies \phi \lor \psi               & &\text{iff} & &s,\alpha \satisfies \phi \text{ or } s,\alpha \satisfies \psi \\
&s, \alpha \satisfies \Ka \phi                    & &\text{iff}  & &s \bowtie \alpha \text{ and } t,\beta \satisfies \phi \text{ for some } (t, \beta) \in W \times\Hist \\
& & & & &\text{ such that } s\sim_at, \alpha \view \beta \text{ and } t\bowtie \beta \\
&s,\alpha\satisfies\langle a\rangle\phi           & &\text{iff} & &|\alpha|_a < |\alpha|_! \text{ and } s,\alpha a\satisfies\phi\\
&s,\alpha\satisfies\langle\phi\rangle\psi         & &\text{iff} & &s,\alpha\satisfies\phi \text{ and } s,\alpha \phi\satisfies\psi \\
\end{align*}
\end{definition}

The semantics for the dual modalities is obtained as usual.

The relation $\satisfies$ is a \emph{satisfaction relation} and $\bowtie$ is an \emph{executability relation}: if $s\bowtie\alpha$, we say that $\alpha$ is executable in $s$.

For the sake of precision, we show that the semantics is well defined before proceeding.

\begin{proposition}
The relations $\bowtie$ and $\satisfies$ are well-defined.
\end{proposition}

\begin{proof}
We distinguish the following cases:
\begin{itemize}
\item $s \bowtie \alpha a$: for any formula $\phi \in \lang_{a}$ we have $(\alpha, \phi) \ll (\alpha a, \phi)$.
Indeed, obviously, $deg(\alpha, \phi) = deg(\alpha a,\phi)$  and $\Vert \alpha \Vert + \Vert \phi \Vert < \Vert \alpha a \Vert + \Vert \phi \Vert = \Vert \alpha \Vert + 1 + \Vert \phi \Vert$.

\item $s \bowtie \alpha \phi$: it is enough to show that, for any formula $\psi \in \lang$, $(\alpha, \psi) \ll (\alpha \phi, \psi)$ and $(\alpha, \phi) \ll (\alpha \phi, \psi)$. Let $\psi$ be a formula in $\lang$.
Obviously, $deg(\alpha, \psi) \leq deg(\alpha \phi,\psi)$ and $\Vert \alpha \Vert + \Vert \psi \Vert < \Vert \alpha \phi \Vert + \Vert \psi \Vert = \Vert \alpha \Vert + \Vert \phi \Vert + \Vert \psi \Vert $ (with $\Vert \phi \Vert > 1)$. Similarly, $deg(\alpha, \phi) \leq deg(\alpha \phi,\psi)$ and $\Vert \alpha \Vert + \Vert \phi \Vert < \Vert \alpha \phi \Vert + \Vert \psi \Vert$.

\item $s,\alpha \satisfies \lnot \phi$: $(\alpha,\phi) \ll (\alpha, \lnot\phi)$.
Indeed, $deg(\alpha, \phi) = deg(\alpha, \lnot \phi)$, and $\Vert \alpha \Vert + \Vert \phi \Vert < \Vert \alpha \Vert + \Vert \lnot \phi \Vert = \Vert \alpha \Vert + \Vert \phi \Vert + 1$.

\item $s,\alpha \satisfies \phi \lor \psi$: $(\alpha,\phi) \ll (\alpha, \phi\lor\psi)$ and $(\alpha,\psi) \ll (\alpha, \phi\lor\psi)$.
\newline Indeed, $deg(\alpha, \phi) \leq deg(\alpha, \phi \lor \psi)$ and $\Vert \alpha \Vert + \Vert \phi \Vert < \Vert \alpha \Vert + \Vert \phi \lor \psi \Vert = \Vert \alpha \Vert + \Vert \phi \Vert + \Vert \psi \Vert$ (and similarly, for $(\alpha, \psi) \ll (\alpha, \phi \lor \psi)$).

\item $s,\alpha \satisfies \Ka \phi$: it is enough to show that $(\beta,\phi) \ll (\alpha, \Ka \phi)$ for all histories $\beta$ such that $\alpha \view \beta$. Let $\beta \in \Words$ be such that $\alpha \view \beta$.
Since $\alpha \view \beta$, $\beta\proj_!$ is a prefix of $\alpha\proj_!$, which implies $ \sum_{\psi \in \beta} deg(\psi) \leq \sum_{\psi \in \alpha} deg(\psi)$. Moreover, by Lemma \ref{lemmaOrder}, $  deg(\langle \beta \rangle \phi)  = \sum_{\psi \in \beta} deg(\psi) + deg(\phi)$ and $ deg(\langle \alpha \rangle \Ka \phi) = \sum_{\psi \in \alpha} deg(\psi) + deg(\Ka \phi) = \sum_{\psi \in \alpha} deg(\psi) + deg(\phi) + 1$. Hence $deg(\beta,\phi) < deg(\alpha, \Ka \phi)$. Therefore $(\beta,\phi) \ll (\alpha,\Ka \phi)$.

\item $s,\alpha \satisfies \langle a \rangle \phi$: $(\alpha a,\phi) \ll (\alpha, \langle a \rangle \phi)$ since $deg(\alpha , \phi) = deg(\alpha, \langle a \rangle \phi)$ and $\Vert \alpha a\Vert + \Vert \phi \Vert = \Vert \alpha \Vert + 1 + \Vert \phi \Vert < \Vert \alpha \Vert + \Vert \phi \Vert + 2 =\Vert \alpha \Vert + \Vert \langle a \rangle \phi \Vert$.

\item $s,\alpha \satisfies \langle \phi \rangle \psi$: we need to show $(\alpha,\phi) \ll (\alpha, \langle \phi \rangle \psi)$ and $(\alpha \phi,\psi) \ll (\alpha, \langle \phi \rangle \psi) $.
On the one hand, $deg(\alpha, \phi) \leq deg(\alpha, \langle \phi \rangle \psi)$ and $\Vert \alpha \Vert + \Vert \phi \Vert < \Vert \alpha \Vert + \Vert \langle \phi \rangle \psi \Vert = \Vert \alpha \Vert + 2\Vert \phi \Vert + \Vert \psi\Vert$ (note that $\Vert \phi \Vert > 1$ for all $\phi \in  \lang$) so $(a,\phi)\ll(a,\langle\phi\rangle\psi)$.
On the other hand, $deg(\alpha \phi, \psi) = deg(\alpha, \langle \phi \rangle \psi)$. Moreover $\Vert \alpha \phi \Vert + \Vert \psi \Vert = \Vert \alpha \Vert + \Vert \phi \Vert + \Vert \psi \Vert < \Vert \alpha \Vert + 2\Vert \phi \Vert + \Vert \psi\Vert = \Vert \alpha \Vert + \Vert \langle \phi \rangle \psi \Vert$. Therefore $(\alpha\phi,\psi) \ll (\alpha, \langle \phi \rangle \psi)$.
\end{itemize}
\end{proof}


Remember that we focus on what agents know based on the announcements they have received and on the histories they consider possible. More precisely, here, agents do not consider histories that cannot be executed in states they consider possible. This is why, in the semantics for knowledge, we consider only pairs ($t,\beta$) where history $\beta$ is indeed executable in state $t$. For convenience, we call \emph{asynchronous epistemic state} (or simply \emph{epistemic state}) a pair of state and word $(s,\alpha)$ such that $s\bowtie\alpha$. Agents, then, only consider possible \emph{asynchronous epistemic states}.

Our satisfaction relation $\satisfies$ is slightly different from that of \cite{AA}: in our case, if $s,\alpha\satisfies \phi$ then $s \bowtie \alpha$, but not so in \cite{AA}. In the Appendix we show that the logics (the sets of validities) are the same. One could say that our semantics is therefore closer to the world deleting semantics of PAL whereas that of \cite{AA} is more akin to the link cutting semantics of \cite{vanBenthem2007preferenceLogic,nomuraPAL}.

Note that equivalence between $s,\alpha \satisfies \lnot \phi$ and $s,\alpha \nvDash \phi$ only holds if $s\bowtie \alpha$. This makes our semantics somewhat three-valued (see Section~\ref{sectionComparison}).

We emphasize that the semantics of announcement $[\phi]\psi$ are different from that in PAL. In PAL, the modality $[\phi]$ combines the effect of broadcasting the formula $\phi$ and synchronized reception by all agents. In our semantics, it only means the broadcasting of the message $\phi$. 


\begin{example}
As in Example \ref{example1}, we consider two agents, Alice and Bob, who are only aware of the truth value of their local variable (respectively $p$ and $q$) and the following sequence of events: message $p \lor q$ is sent, then Alice first receives it, and after that, Bob also receives it. The corresponding states and updates are represented in Figure \ref{fig}.

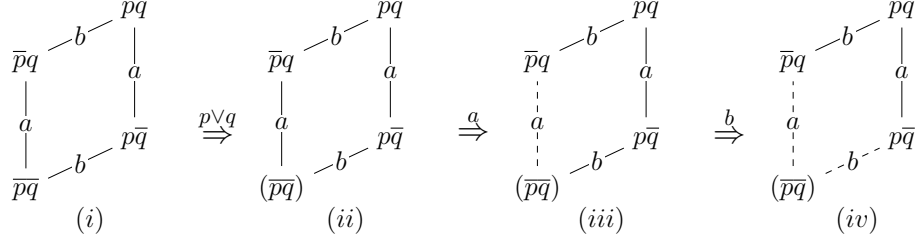
\begin{figure}[h]
\begin{center}
\scalebox{0.85}{
\begin{tikzpicture}
%
\node (00) at (0,0) {$\np\nq$};
\node (01) at (0,2) {$\np\q$};
\node (10) at (1.7,0.8) {$\p\nq$};
\node (11) at (1.7,2.8) {$\p\q$};
\draw (00) -- node[fill=white,inner sep=1pt] {$b$} (10);
\draw (01) -- node[fill=white,inner sep=1pt] {$b$} (11);
\draw (00) -- node[fill=white,inner sep=1pt] {$a$} (01);
\draw (10) -- node[fill=white,inner sep=1pt] {$a$} (11);
\node (t) at (3,1) {\large $\stackrel {p \vel q} \Imp$};
%
\node (00t) at (4,0) {($\np\nq$)};
\node (01t) at (4,2) {$\np\q$};
\node (10t) at (5.7,.8) {$\p\nq$};
\node (11t) at (5.7,2.8) {$\p\q$};
\draw (00t) -- node[fill=white,inner sep=1pt] {$b$} (10t);
\draw (01t) -- node[fill=white,inner sep=1pt] {$b$} (11t);
\draw (00t) -- node[fill=white,inner sep=1pt] {$a$} (01t);
\draw (10t) -- node[fill=white,inner sep=1pt] {$a$} (11t);
\node (t) at (7,1) {\large $\stackrel {a} \Imp$};
\node (00tt) at (8,0) {($\np\nq$)};
\node (01tt) at (8,2) {$\np\q$};
\node (10tt) at (9.7,.8) {$\p\nq$};
\node (11tt) at (9.7,2.8) {$\p\q$};
\draw (00tt) -- node[fill=white,inner sep=1pt] {$b$} (10tt);
\draw (01tt) -- node[fill=white,inner sep=1pt] {$b$} (11tt);
\draw[dashed] (00tt) -- node[fill=white,inner sep=1pt] {$a$} (01tt);
\draw (10tt) -- node[fill=white,inner sep=1pt] {$a$} (11tt);
\node (tt) at (11,1) {\large $\stackrel {b} \Imp$};
\node (00uu) at (12,0) {($\np\nq$)};
\node (01uu) at (12,2) {$\np\q$};
\node (10uu) at (13.7,.8) {$\p\nq$};
\node (11uu) at (13.7,2.8) {$\p\q$};
\draw[dashed] (00uu) -- node[fill=white,inner sep=1pt] {$b$} (10uu);
\draw (01uu) -- node[fill=white,inner sep=1pt] {$b$} (11uu);
\draw[dashed] (00uu) -- node[fill=white,inner sep=1pt] {$a$} (01uu);
\draw (10uu) -- node[fill=white,inner sep=1pt] {$a$} (11uu);
%
\node (i) at (1,-.5) {$(i)$};
\node (ii) at (5,-.5) {$(ii)$};
\node (iii) at (9,-.5) {$(iii)$};
\node (iv) at (13,-.5) {$(iv)$};
\end{tikzpicture}
}
\end{center}
\caption{The announcement $p\vel q$ is sent, after which first Alice and then Bill receives it. States are labelled with the valuations of $p$ and $q$. States that are indistinguishable for an agent are linked with a label for that agent. We omit reflexive arrows.}
\label{fig}
\end{figure}

In $(ii)$ and then $(iii)$ and $(iv)$, we write $(\np \nq)$ and draw dashed lines between that state and $\np q$ (resp. $p\nq$) to represent the update induced by the sending of $p\lor q$ and successive receptions by Alice and Bob. To make that clear, for model $\ModelM$ and states $s,t \in W$, let $(s,\alpha)\sim_a (t,\beta)$ stand for ($s\bowtie\alpha, s\sim_a t, \alpha \view \beta$ and $t \bowtie \beta$).
Furthermore, for any history $\alpha$, we can think of $M^\alpha = (W^\alpha,\sim,V)$ as an \emph{updated model}, where $W^\alpha = \lbrace s \in W \, | \, s\bowtie \alpha \rbrace$.

After $p\lor q$ is sent, $\np\nq$ does not belong to the updated model anymore because it does not satisfy $p \lor q$: $\np\nq \not\bowtie p \lor q$. However, since neither Alice nor Bob has received the message yet, $\np\nq$ is still accessible to them and therefore we should keep it in the model representation. Indeed, in $(ii)$, $(p\nq,p \lor q)\sim_a (\np\nq,\epsilon)$ so Alice still considers $\np\nq$ as a possible state (and so does Bob).
Similarly, in $(iv)$, although both Alice and Bob have received message $p \lor q$, they do not know that the other also has, hence they both consider possible that the other still thinks that $\np\nq$ is a possible state.
Hence, even then, Bob, for instance, considers possible that Alice still thinks the current model is that depicted in $(i)$, because $(t,p \lor q.a.b)\sim_b (t,p \lor q.b)\sim_a(s,\epsilon)$.
In that sense, what can be seen as an update of the model is not commonly known amongst the agents.

The executability relation, however, assumes the role of a witness for the update because $\np\nq \not\bowtie p \lor q$ means that state $\np\nq$ does not survive the sending of $p \lor q$.

\end{example}

\bigskip

From the definition of the semantics, two different notions of validity arise.

\begin{definition}[Validity]\label{defValidity}
We define two distinct notions of validity as follows:
\begin{itemize}
\item Validity with \emph{empty word}: $\phi$ is \emph{$\epsilon$-valid} (or \emph{valid}) if and only if for all models $\ModelM$ and all states $s\in W$, $s,\epsilon \satisfies \phi$.
\item Validity with \emph{arbitrary word}: $\phi$ is \emph{$\ast$-valid} (or \emph{always valid}) if and only if for all words $\alpha$, $[\alpha]\phi$ is valid.
\end{itemize}
The set of all $\epsilon$-validities is called AA$^\epsilon$ (or simply AA) and the set of all $\ast$-validities is called AA$^\ast$.
\end{definition}

From those definitions, the following proposition is obvious.
\begin{proposition}[$\ast$-validity implies $\epsilon$-validity]\label{propAlwaysValidImpliesValid}
Let $\phi \in \lang$. If $\satisfies^\ast \phi$ then $\satisfies \phi$.
\end{proposition} 

\noindent Note that the converse does not hold: there is a formula $\phi$ such that $\phi$ is $\epsilon$-valid but not $\ast$-valid, \emph{i.e.} there is a word $\alpha$ such that $\nvDash [\alpha]\phi$. For instance, $\phi = [a]\F$ is $\epsilon$-valid but not $\ast$-valid since $\nvDash[\T][a]\F$.

\subsection{Some properties of the semantics}

We continue with some interesting properties of the satisfaction and the executability relations. In the following, we consider an arbitrary model $M=(W,\sim,V)$.


\begin{lemma}\label{lemmaBowtieHistory}
Let $s$ be a state. For any word $\alpha$, if $s\bowtie \alpha$ then $\alpha$ is a history.
\end{lemma}

\begin{proof}
The proof proceeds by straightforward induction on $\Vert \alpha \Vert$.
\end{proof}

%
%


\begin{lemma}\label{lemmaExecutabilityPrefix}
Let $\alpha,\beta$ be words and $s$ be a state.
If $s\bowtie \alpha$ and $\beta \prefix \alpha $ then $s\bowtie \beta$.
\end{lemma}

\begin{proof}
	By straightforward induction on $\Vert \alpha \Vert$.
\end{proof}


\noindent Note that the executability relation $\bowtie$ does not only check whether $\alpha$ is a history but it also verifies that $\alpha$ is \emph{consistent} and \emph{executable} in a given state. For instance, obviously $\alpha := p. b. \lnot K_b\p$ is a history. However, it cannot be executed in any state because such a history is somehow inconsistent: if agent $b$ has received information $p$ then $b$ knows that $p$, so $\lnot K_b p$ is false and thus cannot be announced. Hence, for any model $\ModelM$ and any state $s\in W$, $s\not\bowtie\alpha$ because $s,p. b \nvDash \lnot K_b p$. This relates to the notion of consistent cut in distributed computing, to which we will go back in Section \ref{sectionComparison}. \\

The following proposition states that whenever a pair $(s,\alpha) \in W\times\Words$ satisfies a formula $\phi$, word $\alpha$ is indeed an executable  history in state $s$, \emph{i.e.} $s\bowtie \alpha$.

\begin{proposition}\label{propExecutabilityIfSatisfaction}
Let $s \in W$ be a state, $\alpha \in \Words $ a word and $\phi \in \lang$ a formula. If $s,\alpha \satisfies \phi$ then $s \bowtie \alpha$.
\end{proposition}

\begin{proof}
The proof proceeds by $\ll$-induction on $(\alpha, \phi)$. This is straightforward.
\end{proof}

\noindent Therefore, if $s,\alpha \satisfies \phi$ for some formula $\phi \in \lang$, the pair $(s,\alpha)$ is an epistemic asynchronous state and $\alpha$ is a history.


We further give some results that will be helpful for future proofs.

\begin{lemma}\label{LemmaSemanticsDiamondHistory}
For any state $s$, any words $\alpha, \beta$ and any formula $\phi$:
\begin{align*}
(1) \qquad &s,\alpha \satisfies \langle \beta \rangle \phi & &\text{if and only if} & &s\bowtie \alpha \beta \text{ and } s,\alpha \beta \satisfies \phi\\
(2) \qquad &s,\alpha \satisfies [\beta]\phi & &\text{if and only if} & &s,\alpha \satisfies \lnot \langle \beta \rangle \lnot\phi
\end{align*}
\end{lemma}

\begin{proof}
Let $\alpha$ be a word and $s$ a state. The proof proceeds by induction on $\Vert \beta \Vert$. Let $\beta$ be a word such that for all words $\gamma$, if $\Vert\gamma\Vert < \Vert \beta \Vert$ then for all formulas $\phi$, $(1) \ s,\alpha \satisfies \langle \gamma \rangle \phi$ if and only if $s\bowtie \alpha\gamma$ and $s,\alpha\gamma \satisfies \phi$ and $(2) \ s,\alpha \satisfies [\beta]\phi$ if and only if $s,\alpha\satisfies \lnot\langle\beta\rangle\lnot\phi$. We show that those properties also holds for $\beta$. We distinguish three cases:
\begin{itemize}
\item Case $\epsilon$: (1) by Proposition \ref{propExecutabilityIfSatisfaction}, $s,\alpha \satisfies \phi$ implies $s\bowtie \alpha$ so, obvisouly, $s,\alpha \satisfies \phi$ if and only if $s\bowtie \alpha$ and $s,\alpha \satisfies \phi$. (2) Obviously $s,\alpha \satisfies \phi \Eq s,\alpha \satisfies \lnot \lnot \phi$.
\item Case $\beta a$: the inductive hypothesis (IH) applies because $\Vert \beta \Vert < \Vert \beta a\Vert$. Concerning (1) we have the following equivalences
\begin{align*}
s,\alpha \satisfies \langle \beta a \rangle \phi
&\Eq s,\alpha \satisfies \langle \beta \rangle \langle a \rangle \phi &\\
&\Eq s \bowtie \alpha \beta \text{ and } s,\alpha \beta \satisfies \langle a \rangle \phi & &\text{ by (IH)}\\
&\Eq s \bowtie \alpha \beta \text{ and } |\alpha\beta|_a < |\alpha\beta|_! \text{ and } s,\alpha \beta a \satisfies \phi & \\
&\Eq s\bowtie \alpha \beta a \text{ and } s,\alpha \beta a \satisfies \phi &
\end{align*}
and concerning (2):
\begin{align*}
s,\alpha \satisfies [\beta a]\phi &\Eq s,\alpha \satisfies [\beta][a]\phi &\\
&\Eq s,\alpha \satisfies \lnot \langle \beta \rangle \lnot [a]\phi  & &\text{by (IH)}\\
&\Eq s,\alpha \satisfies \lnot \langle \beta \rangle \langle a\rangle \lnot \phi & &\text{by definition}\\
&\Eq s,\alpha \satisfies \lnot \langle \beta a \rangle \lnot \phi &
\end{align*}

\item Case $\beta \psi$: here the induction hypothesis applies because $\Vert \beta \Vert < \Vert \beta \psi \Vert$. For (1), we have
\begin{align*}
s,\alpha \satisfies \langle \beta \psi \rangle \phi
&\Eq s,\alpha \satisfies \langle \beta \rangle \langle \psi \rangle \phi & \\
&\Eq s \bowtie \alpha \beta \text{ and } s,\alpha \beta \satisfies \langle \psi \rangle \phi & &\text{ by (IH)} \\
&\Eq s \bowtie \alpha \beta \text{ and } s,\alpha \beta \satisfies \psi \text{ and } s,\alpha \beta \psi \satisfies \phi & \\
&\Eq s\bowtie \alpha \beta \psi \text{ and } s,\alpha \beta \psi \satisfies \phi& 
\end{align*}
and concerning (2):
\begin{align*}
s,\alpha \satisfies [\beta \psi]\phi &\Eq s,\alpha \satisfies [\beta][\psi]\phi &\\
&\Eq s,\alpha \satisfies \lnot \langle \beta \rangle \lnot [\psi]\phi  & &\text{by (IH)} \\
&\Eq s,\alpha \satisfies \lnot \langle \beta \rangle \langle \psi\rangle \lnot \phi & &\text{by definition}\\
&\Eq s,\alpha \satisfies \lnot \langle \beta \psi \rangle \lnot \phi
\end{align*}
\end{itemize}
\end{proof}

\begin{corollary}\label{coSemanticsDiamondHistory}
For any state $s$ and any words $\alpha,\beta$:
\begin{align*}
(1) \qquad &s,\alpha \satisfies \langle \beta \rangle \T & &\text{if and only if} & &s\bowtie \alpha \beta\\
(2) \qquad &s,\alpha \satisfies [\beta] \F  & &\text{if and only if} & &s,\alpha \satisfies \lnot \langle \beta \rangle \T
\end{align*}
Moreover, if $s\bowtie \alpha$, then $s,\alpha \satisfies \langle \beta \rangle \T$ if and only if $s,\alpha \nvDash [\beta]\F$.
\end{corollary}

We can now prove:

\begin{proposition}\label{PropExecutabilityHistoryNotFalse}
For any state $s$ and any words $\alpha, \beta$, if $s\bowtie \alpha$, then $s\bowtie \alpha\beta$ if and only if $s,\alpha \nvDash [\beta]\F$.
\end{proposition}

\begin{proof}
Suppose $s\bowtie\alpha$.
If $s \bowtie \alpha \beta$ then, by definition, $s,\alpha \beta \satisfies \T$. Now, from Lemma \ref{LemmaSemanticsDiamondHistory} we get $s,\alpha \satisfies \langle \beta \rangle \T$ and, since $s\bowtie\alpha$, we conclude from Corollary \ref{coSemanticsDiamondHistory} that $s,\alpha \nvDash [\beta]\F$.
Conversely, if $s,\alpha \nvDash [\beta]\F$, then, by Corollary \ref{coSemanticsDiamondHistory}, $s,\alpha \satisfies \langle \beta \rangle \T$ so $s \bowtie \alpha \beta$. 
\end{proof}

Now we can obtain semantic definitions for dynamic modalities with words:

\begin{lemma}
Let $\alpha,\beta$ be two words over $\A \union \lang$ and $\phi \in \lang$ a formula. For all states $s$,
\begin{align*}
(1) \quad s,\alpha \satisfies \langle \beta \rangle \phi \quad &\text{iff} \quad s\bowtie\alpha\beta \text{ and } s,\alpha\beta \satisfies \phi \\
(2) \hspace{1em} s,\alpha \satisfies \hspace{0.2em} [\beta] \phi \quad &\text{iff} \quad s\bowtie\alpha \text{ and, if } s\bowtie \alpha\beta, \text{ then } s,\alpha\beta \satisfies \phi
\end{align*}
\end{lemma}

\begin{proof}
(1) is obtained directly from Lemma \ref{LemmaSemanticsDiamondHistory}. For (2), we also use the fact that $s,\alpha \satisfies [\beta] \phi$ if and only if $s,\alpha \satisfies \lnot \langle \beta\rangle \lnot \phi$, by Lemma \ref{LemmaSemanticsDiamondHistory}.
\end{proof}

\begin{corollary}
Let $\beta$ be a word over $\A \union \lang$ and $\phi \in \lang$ be a formula. For all states $s$,
\begin{align*}
(1) \quad s,\epsilon \satisfies \langle \beta \rangle \phi \quad &\text{iff} \quad s\bowtie\beta \text{ and } s,\beta \satisfies \phi \\
(2) \hspace{1em} s,\epsilon \satisfies \hspace{0.2em} [\beta] \phi \quad &\text{iff} \quad \text{if } s\bowtie \beta, \text{ then } s,\beta \satisfies \phi
\end{align*}
\end{corollary}

This way we obtain an alternative definition for $\ast$-validities:

\begin{corollary}
For all $\phi \in \lang$, $\satisfies^\ast \phi$ if and only if for all models $\ModelM$, for all states $s \in W$ and for all words $\alpha \in \Words$ such that $s\bowtie \alpha$, $s,\alpha \satisfies \phi$.
\end{corollary}

\section{Axiomatisation $\AAstar$} \label{sectionAxiomatisation}

In this section we propose an axiomatisation for the set of always-validities AA$^\ast$. Our axiomatisation is not based on reduction axioms but displays a sort of reduction from $\ast$-validities to $\epsilon$-validities.

\paragraph*{Comparison of AA and AA$^\ast$.}

As expected, axiomatisation {\bf AA}$^\ast$ does not extend \textbf{AA} (see \cite{AA}), but should rather be seen as a restriction of {\bf AA}. Indeed the set of $\ast$-validities is included in that of $\epsilon$-validities (see Proposition \ref{propAlwaysValidImpliesValid}). The {\bf AA} axioms (A3): $[\alpha a]\F$ if $|\alpha|_a \geq |\alpha|_!$, and (A7): $[\alpha]K_a \phi \leftrightarrow [\alpha]\F \lor \bigwedge_{\alpha \view \beta} K_a [\beta]\phi$ are not $\ast$-valid.
To show that $\nvDash^\ast[\alpha a]\F$ if $|\alpha|_a \geq |\alpha|_!$, take $\beta= p. q$ and $\alpha =a$. Let $M=(W,\sim,V)$ be a model and $s\in W$ a state such that $s\bowtie \beta$, \emph{i.e.} $s\in V(p)$ and $s,\in V(q)$. Then $s,\beta \nvDash [\alpha  a]\F $ because $s\bowtie p. q. a. a$.

For $[\alpha]K_a \phi \leftrightarrow [\alpha]\F \lor \bigwedge_{\alpha \view \beta} K_a [\beta]\phi$, we give a counterexample in the single-agent case $\A = \lbrace a \rbrace$. Consider a model $M=(W,\sim,V)$ with two states $s,t \in W$ such that $s\sim_a t$. Suppose $p$ is true in $s$ and $t$ but $q$ is only valid in $s$. Let $ \beta= p. p. a$ and $\alpha = q. a$. We consider the formula $\phi := \lnot K_a q$. Note that $s\bowtie \beta \alpha$ so $s,\beta \nvDash [\alpha]\F$.

\begin{figure}[H]
\centering
%
%

\begin{tikzpicture}
\node (00) at (0,0) {$pq$};
\node (01) at (2,0) {$p\nq$};

\draw[<->] (00) --  (01);
\draw [->] (00) edge[loop above] (00);
\draw [->] (01) edge[loop above] (01);
\end{tikzpicture}

\end{figure}

\noindent Here $s,\beta \satisfies [\alpha] K_a  \phi$ because $s,p. p. a. q. a \satisfies K_a \lnot K_a q$. Indeed, \textbf{view}$_a(p . p. a. q. a)= \lbrace p. p. a. a \rbrace$ and $s,p. p. a. a \satisfies \lnot K_a q$ and $t,p. p. a. a \satisfies \lnot K_a q$ since $t,p. p. a. a \satisfies \lnot q$.
However $s,\beta \nvDash \bigwedge_{\alpha \view \gamma} K_a [\gamma]\lnot K_a q$ because $s,\beta \nvDash K_a [q. a]\lnot K_a q$. In fact we have $s,p. p. a \satisfies K_a [q. a]K_a q$. This is the case because \textbf{view}$_a(p. p. a)= \lbrace p. a \rbrace$, $s,p. a \satisfies [q. a]K_a q$ and trivially $t,p. a \satisfies [q. a]K_a q$ (because $t \not\bowtie p. a. q. a$).
Therefore, $ s,\beta \nvDash [\alpha]K_a \phi \rightarrow [\alpha]\F \lor \bigwedge_{\alpha \view \gamma}K_a[\gamma]\phi$.

\paragraph*{Always-validities cannot be eliminated.} Axiomatisation \textbf{AA$^\ast$} is not a reduction system because dynamic modalities cannot be eliminated from $\ast$-validities. Consider the formula $[a]\F \in \lang$. Suppose towards a contradiction that there is a formula $\phi \in \lang_{ml}$ without any dynamic modalities such that $\satisfies^\ast [a]\F \eq \phi$. Since $\satisfies [a]\F$, also $\satisfies \phi$. But now, $\epsilon$-validity and $\ast$-validity coincide for any formula in the language of basic modal logic. Hence $\satisfies^\ast\phi$. However, $\nvDash ^\ast [a]\F$ because obviously $\nvDash [\T][a]\F$. Therefore $\nvDash^\ast [a]\F \eq \phi$. This shows that reduction axioms cannot provide a complete axiomatisation for $\ast$-validities.

However, we can somehow reduce $\ast$-validities to $\epsilon$-validities although in an infinitary way. Indeed, for any formula $\phi \in \lang$, $\satisfies^\ast \phi$ if and only if $\satisfies [\alpha]\phi$ for all words $\alpha$. Then, showing that $\phi$ is $\ast$-valid boils down to showing that for all words $\alpha$, $[\alpha]\phi$ is $\epsilon$-valid. The intuitive idea is to go back from $\alpha$ to the initial empty history $\epsilon$ and analyse the formulas from there. In the following, we show how we can use such an idea to provide a complete axiomatisation for AA$^\ast$.\\

From now on, let $\A$ be a set with at least two agents---that this assumption is necessary will be explained later on. We define the following formula which is meant to express the fact that the current history is empty:
\begin{displaymath}
\emp := \bigwedge_{a\in \A}[a]\F \land \bigwedge_{a,b\in \A} K_a[b]\F.
\end{displaymath}

The first conjunct of this formula says that all agents have received all sent messages---namely none---and the second conjunct that all agents know this. Let us explain why both are necessary to express that the history is empty. In the following, we consider a state $s\in V(p)$ in a given model $M$. The first conjunct characterises that, for every agent $a\in \A$, $|\alpha|_!=|\alpha|_a$: this is not enough to enforce the history to be empty, as given two agents $a$ and $b$, $s,p.a.b\satisfies [a]\bot \land [b]\bot$. By itself, the second conjunct is not sufficient either since, \emph{e.g.}, $s,p\satisfies K_a([a]\bot \land [b]\bot)\land K_b([a]\bot \land [b]\bot)$: whenever she has not received any message herself, an agent always knows---or rather \emph{thinks}---there is nothing to receive so that all agents have received all messages. 
The history is only empty if both conjuncts are true.

\begin{lemma}\label{lemmaPiEpsilonWord}
For any model $(W,\sim,V)$, any state $s\in W$ and for all words $\alpha$ over $\lang\cup \A$, $s,\epsilon \satisfies \langle \alpha\rangle\emp$ if, and only, if $\alpha = \epsilon$.  Consequently $s,\alpha \satisfies \emp$ if and only if $\alpha = \epsilon$.
\end{lemma}

\begin{proof}
Let $(W,R,V)$ be a model and $s\in W$ be a state. Let $\alpha$ be a word.

Suppose $s,\epsilon \satisfies \langle \alpha \rangle \emp$. Then $s\bowtie \alpha$ and $s,\alpha \satisfies \emp$. Hence, by Proposition \ref{lemmaBowtieHistory}, $\alpha$  is a history. Now suppose, towards a contradiction, that $\alpha \neq \epsilon$. Then $\alpha = \alpha' \phi$ or $\alpha = \alpha' a$ for some formula $\phi \in \lang$ or some agent $a\in\A$. 
If $\alpha = \alpha'\phi$ then, since $s\bowtie \alpha$, $\alpha'\phi$ and also $\alpha'$ are histories, so in particular, for any $a\in \A$, $|\alpha'|_a \leq |\alpha'|_!$. Hence $|\alpha'\phi|_a < |\alpha'\phi|_!$ for all $a\in\A$, so $s\bowtie\alpha'\phi a$ for all $a\in\A$. Therefore $s,\alpha \nvDash [a]\F$\footnote{We even have $s,\alpha \satisfies \bigwedge_{a\in \A}[a]\T$. Here, $\phi$ can be seen as an unread formula.}. Hence, $ s,\alpha \nvDash \bigwedge_{a\in \A}[a]\F$. Therefore $s,\alpha \nvDash \emp$.
If $\alpha = \alpha'a$, since $s,\alpha \satisfies \emp$, in particular $ s,\alpha' a \satisfies \bigwedge_{c\in \A}[c]\F$ so $s,\alpha' a \satisfies [a]\F$. From this and the fact that $\alpha' a$ is a history we get $|\alpha' a |_a = |\alpha' a |_!$. Hence $|\alpha'|_a = |\alpha' a|_a -1 = |\alpha'|_! -1$ so $s,\alpha' \nvDash [a]\F$. Now, consider another agent $b\in \A, b\neq a$. Note that, by Definition \ref{defViewRelation}, $\alpha'a \tri_b \alpha'$, because $(\alpha' a )\proj_{!b}= \alpha' \proj_! = \alpha' \proj_{!a}$. Moreover, $s\sim_b s$ and $s,\alpha' \nvDash [a]\F$. Also, since $s,\alpha' a \satisfies \emp$, by Proposition \ref{propExecutabilityIfSatisfaction}, $s\bowtie \alpha' a$ and then $s\bowtie \alpha'$, by Proposition \ref{lemmaExecutabilityPrefix}. Therefore $s,\alpha'a \nvDash K_b[a]\F$. So $ s,\alpha \nvDash \bigwedge_{c,d\in\A}K_c[d]\F$. Hence $s,\alpha \nvDash \emp$.
In both cases we get a contradiction. Therefore, $\alpha = \epsilon$.

Conversely, if $\alpha = \epsilon$ then, obviously, $ \ s,\epsilon \satisfies \bigwedge_{a\in \A}[a]\F$. Now, let $a,b\in \A$. Since \textbf{view}$_a(\epsilon) = \lbrace \epsilon \rbrace$, we have $t,\epsilon \satisfies [b]\F$ for all states $t$ such that $s\sim_at$. Therefore $s,\epsilon \satisfies K_a[b]\F$. As agents $a,b$ were arbitrary, we conclude that $ s,\epsilon \satisfies \bigwedge_{a,b\in \A} K_a[b]\F$. Hence $s,\epsilon \satisfies \emp$.
\end{proof}

\begin{corollary}\label{corollaryPiValid}
$\satisfies \emp$.
\end{corollary}

From Lemma \ref{lemmaPiEpsilonWord} and the definition of the view relation, it is also easy to show that whenever the history is empty, every agent knows it. This is due to our assumption that agents do not imagine histories with announcements they have not yet received.

\begin{corollary}\label{corollaryKPiValid}
    $\satisfies \emp \imp K_a \emp$.
\end{corollary}

\begin{proof}
    Let us consider a model $(W,\sim,V)$, a state $s\in W$ and a word $\alpha$. Suppose $s,\alpha \satisfies \emp$. By Lemma \ref{lemmaPiEpsilonWord}, this implies $\alpha = \epsilon$. Now consider a state $t$ and a history $\beta$ such that $s\sim_a t$ and $\alpha \view \beta$. By Definition \ref{defViewRelation}, \textbf{view}$_a(\epsilon)=\lbrace \epsilon \rbrace$ so $\beta = \epsilon$. Then, by Lemma \ref{lemmaPiEpsilonWord} again, $t,\beta \satisfies \emp$. Therefore $s,\alpha \satisfies K_a \emp$.
\end{proof}

\noindent With only one agent, Lemma \ref{lemmaPiEpsilonWord} does not hold anymore. Indeed, an agent $a$ always knows\footnote{Here it would be more precise to say that any agent always `believes' she has read all messages but we prefer to stick to the notion of knowledge. See Section \ref{sectionBorK} for motivation.} that she has read all messages: for all words $\alpha$, $[\alpha] K_a[a]\F$ is valid. Indeed: suppose there are $s$ a state and $\alpha$ a word such that $s\bowtie \alpha$ but $s,\alpha \nvDash K_a [a]\F$. That means there are a state $t$ and a history $\beta$ such that $s\sim_a t, \alpha \view \beta, t\bowtie \beta$, and $t,\beta \nvDash [a]\F$. This implies $|\beta|_a < |\beta|_!$. But, by Definition \ref{defViewRelation}, since $\alpha \view \beta$, $|\beta|_a = |\beta|_!$. This yields a contradiction. Therefore, $\satisfies [\alpha]K_a[a]\phi$ for all words $\alpha$. Hence $K_a[a]\F$ is $\ast$-valid. Therefore, for any history $\alpha$ such that $|\alpha|_a = |\alpha|_!$ and any state $s$ such that $s\bowtie\alpha$, $s,\alpha \satisfies [a]\F \land K_a[a]\F$. For instance, if $s\in V(p)$, $s,p.a\satisfies [a]\F \land K_a[a]\F$. The single-agent case will be discussed in Section \ref{sectionAxiomOneAgent}.\\

The following lemma offers a first link between validities and always-validities, showing how formula $\emp$ can be used to reduce $\ast$-validities to $\epsilon$-validities.

\begin{lemma}\label{lemmaRelationPiValidities}
For all formulas $\phi \in \lang$, $\satisfies^\ast \emp \imp \phi$ if and only if $\satisfies \phi$.
\end{lemma}

\begin{proof}
Let $\phi \in \lang$ be a formula.
Suppose $\satisfies^\ast \emp \imp \phi$. This means for all models $\ModelM$, states $s\in W$ and words $\alpha$, $s,\epsilon \satisfies [\alpha](\emp \imp \phi)$. In particular, for all models $\ModelM$ and states $s\in W$, $s,\epsilon \satisfies \emp \imp \phi$, so $s,\epsilon \satisfies \phi$ because, by Corollary \ref{corollaryPiValid}, $s,\epsilon \satisfies \emp$. Hence $\satisfies \phi$.

Conversely, suppose $\satisfies \phi$. Let $\ModelM$ be a model, $s\in W$ a state and $\alpha$ a word. Suppose $s\bowtie \alpha$ and $s,\alpha \satisfies \emp$. Then, by Lemma \ref{lemmaPiEpsilonWord}, $\alpha = \epsilon$. Now $s,\epsilon \satisfies \phi$ by hypothesis so $s,\epsilon \satisfies \emp \imp \phi$. Since $\alpha = \epsilon$ and $[\epsilon]\psi = \psi$ by definition (for all formula $\psi$), $s,\epsilon \satisfies [\alpha](\emp \imp \phi)$. Therefore $\satisfies^\ast \emp \imp \phi$.
\end{proof}

Finally, the following proposition establishes the kind of reduction from $\ast$-validities to $\epsilon$-validities that we need to axiomatise AA$^\ast$.

\begin{proposition}\label{propLinkValiditiesforAxiomatisation}
\begin{align*}
\satisfies^\ast \phi \quad &\Eq \quad \satisfies [\alpha]\phi & &\text{for all words } \alpha \\
\quad &\Eq \quad \satisfies^\ast \emp \imp [\alpha]\phi & &\text{for all words } \alpha
\end{align*}
\end{proposition}

\begin{proof} Let $\phi$ be a formula. By Definition \ref{defValidity}, $\satisfies^\ast \phi$ if and only if $\satisfies[\alpha]\phi$ for all words $\alpha$. Moreover, for all words $\alpha$, $\satisfies [\alpha]\phi $ if and only if $\satisfies ^\ast \emp \imp [\alpha]\phi$, by Lemma \ref{lemmaRelationPiValidities}.
\end{proof}

We have now all the ingredients to present our axiomatisation for AA$^\ast$.


\begin{definition}[Axiomatisation AA$^\ast$]
The axiomatisation \textbf{AA$^\ast$} is composed of the axioms and rules in Table \ref{AAstar}.

\begin{table}[h]
\begin{center}
\begin{tabular}{|ll|}
\hline
          		 & all instances of tautologies \\
(Dist)    		 & $K_a(\phi \imp \psi) \imp (K_a\phi \imp K_a\psi)$ 	 \\
(Dist!)   		 & $[\alpha](\phi \imp \psi) \imp ([\alpha]\phi \imp [\alpha]\psi)$  \\	  
($\emp$K)  		 & $\emp \imp K_a \emp$  	\\								   
($\emp$T)  		 & $\emp \imp (K_a\phi \imp \phi)$  	\\
($4$)   	     & $K_a\phi \imp K_aK_a\phi$ \\   
($5$)  	         & $\Ka\phi \imp K_a\Ka \phi$ 		 \\
(Exec!$_1$)   	 & $\langle \phi \rangle \T \eq \phi$  \\
(Exec!$_2$)    	 & $[\alpha]\langle a \rangle \T$  \hspace{4.9em} if  $|\alpha|_a < |\alpha|_!$  \\
(Exec!$_3$)   	 & $\emp \imp [\alpha][a]\F$ \hspace{1em} if $|\alpha|_a \geq |\alpha|_!$  \\

(Func!)    	 	 & $\langle \alpha \rangle \phi \imp [\alpha]\phi$ \\
(Perm!)     	 & $(p \imp [\alpha]p) \land (\lnot p \imp [\alpha]\lnot p) $  \\
($\emp$!)        &  $\displaystyle \emp \imp \Bigl([\alpha]K_a\phi \eq \Bigl( [\alpha]\F \lor \bigwedge_{\alpha \view \beta} K_a[\beta]\phi \Bigr) \Bigr)$ \\
(MP)     		 & from $\phi$ and $\phi \imp \psi$, infer $\psi$ \\
(NecK)     		 & from $\phi$, infer $K_a\phi$	\\
(Nec!)     		 & from $\phi$, infer $[\alpha]\phi$ \\
(R$^\ast$) 		 &  from $\emp \imp [\alpha]\phi$ for all words $\alpha$, infer $\phi$  \\

\hline
\end{tabular}
\end{center}
\caption{Axiomatisation \textbf{AA$^\ast$}}\label{AAstar}
\end{table}
\label{tabletwo}
\end{definition}


As for epistemic modalities, we have a distribution axiom and a necessitation rule for dynamic modalities. Axiom T ($K_a \phi \imp \phi$) for knowledge factivity is here prefixed by $\emp$ because epistemic modalities display equivalence relations in the initial model---\emph{i.e.} with the empty history---only. Indeed, while axioms 4 and 5 are both $\ast$-valid, axiom T (as well as axiom B ($p \imp K_a \Ka p$) that corresponds to symmetry) is only $\epsilon$-valid.
Axioms (Exec!$_i$) express executability conditions for histories and represent the properties of the agreement relation. In particular, (Exec!$_1$) represents the constraint that a formula can be announced if and only if it is true; and (Exec!$_3$) expresses the fact that any word that is not a history, here such that $|\alpha|_a > |\alpha|_!$, is not executable. (Perm!) corresponds to atomic permanence: announcements do not change the atoms truth value.
Finally, ($\emp$!) is obtained from \textbf{AA} by prefixing axiom (A7): $[\alpha]K_a \phi \leftrightarrow [\alpha]\F \lor \bigwedge_{\alpha \view \beta} K_a [\beta]\phi$ with $\emp$: this is a typical case where $\ast$-validities are reduced to $\epsilon$-validities through the use of formula $\emp$.
As for the rules, the so-called \emph{modus ponens} (MP) and the necessitation rules (NecK) and (Nec!) are as expected, while (R$^\ast$) is an infinitary rule that provides the first necessary step towards reducing $\ast$-validities to $\epsilon$-validities, following Proposition \ref{propLinkValiditiesforAxiomatisation}.

Note that by definition $[\alpha]\phi := \lnot \langle \alpha \rangle \lnot \phi$ so obviously $[\alpha]\phi \eq \lnot \langle \alpha \rangle \lnot \phi \in \AAstar$. Then also $\langle \alpha \rangle \phi \eq \lnot [\alpha]\lnot \phi \in \AAstar$ by necessitation, distribution and basic propositional reasoning, so we do not need to add a specific axiom for duality of the dynamic modalities.

\begin{example}
We show that $[\alpha]\phi \imp ([\alpha]\F \lor \langle \alpha \rangle \phi)$ is a theorem of $\AAstar$:

\begin{align*}
    (1) \quad &\phi \imp (\lnot \phi \imp \F) 								        & &\text{propositional tautology}\\
    (2) \quad &[\alpha](\phi \imp (\lnot \phi \imp \F))    					        & &\text{by Necessitation } (1) \\
    (3) \quad &[\alpha]\phi \imp [\alpha](\lnot \phi \imp \F)  				        & &\text{(Dist!) and Modus Ponens } (2)\\
    (4) \quad &[\alpha]\phi \imp ([\alpha]\lnot \phi \imp [\alpha]\F)		        & &\text{(Dist!), propositional reasoning } (3) \\
    (5) \quad &[\alpha]\phi \imp (\lnot [\alpha]\lnot\phi \lor [\alpha]\F)	        & &\text{propositional reasoning } (4) \\
    (6) \quad &[\alpha]\phi \imp (\langle \alpha \rangle \phi \lor [\alpha]\F)		& &\text{by definition of }[\cdot]
\end{align*}

It is also easy to show that $(\langle \alpha \rangle \phi \lor [\alpha]\F) \imp [\alpha]\phi \in \AAstar$:
\begin{align*}
    (1) \quad &\langle \alpha \rangle \phi \imp [\alpha]\phi                    &
    &\text{(Func!)} \\
    (2) \quad &\F \imp \phi                                                     &
    &\text{proposition tautology}\\
    (3) \quad &[\alpha](\F \imp \phi)                                           &
    &\text{(Nec!) } (2)\\
    (4) \quad &[\alpha]\F \imp [\alpha]\phi                                     &
    &\text{(Dist!) and Modus Ponens }(3)\\
    (5) \quad &([\alpha]\bot \lor \langle \alpha \rangle \phi) \imp[\alpha]\phi &
    &\text{propositional reasoning } (1),(4)
\end{align*}

Therefore, $[\alpha]\phi \eq ([\alpha]\F \lor \langle \alpha \rangle \phi)$ is a theorem of $\AAstar$.
\end{example}

\begin{example}
\label{propDerivabilityKaAllRead}
As an example of how the (R$^\ast$) rule is used, we show that $\AAstar \vdash K_a[a]\F$.

To show that $K_a[a]\F$ is a theorem, we show that $\emp \imp [\alpha]K_a [a]\F$ is derivable, for all words $\alpha$. Let $\alpha$ be a word. We have:
\begin{align*}
(1) \ &\emp \imp [\beta][a]\F & &\text{for all $\alpha \view \beta$, by } (Exec!_3)\text{, since } |\beta|_a = |\beta|_! \\
(2) \ &K_a (\emp \imp [\beta][a]\F) & &\text{for all $\alpha \view \beta$, by } (NecK) (1) \\
(3) \ &K_a \emp \imp K_a[\beta][a]\F & &\text{for all $\alpha \view \beta$, by } (DistK) (2) \\
(4) \ &\emp \imp K_a \emp & &\text{by } (\emp K) \\
(5) \ &\emp \imp K_a [\beta][a]\F & &\text{for all $\alpha \view \beta$ by propositional reasoning } (3,4) \\
(6) \ &\emp \imp \bigwedge_{\alpha \view \beta} K_a [\beta][a]\F & &\text{by propositional reasoning } (5) \\
(7) \ &\emp \imp \Bigl( [\alpha]\F \lor \bigwedge_{\alpha \view \beta} K_a [\beta][a]\F \Bigr)  & &\text{by propositional reasoning } (6) \\
(8) \ &\emp \imp \Bigl( [\alpha]K_a [a]\F \eq \Bigl( [\alpha]\F \lor \! \bigwedge_{\alpha \view \beta} \! K_a [\beta][a]\F \Bigr) \Bigr) & &\text{by } (\emp!) \\
(9) \ &\emp \imp [\alpha]K_a [a]\F & &\text{by propositional reasoning } (7,8)
\end{align*}

\noindent Therefore, for all words $\alpha$, $\emp \imp [\alpha]K_a[a]\F$ is a theorem. By rule (R$^\ast$), then, $K_a[a]\F$ is a theorem.

\end{example}

\hfill \\

We now state the main results about this axiomatisation, namely its soundness and completeness.

\begin{theorem}[Soundness]\label{Soundness}
$\AAstar$ is sound w.r.t. AA$^\ast$ \emph{i.e.} if $\phi \in \textbf{AA}^\ast$ then $\satisfies^\ast \phi$.
\end{theorem}

\begin{proof}
It is enough to show that all axioms are $\ast$-valid and that all rules preserve $\ast$-validity.
Corollary \ref{corollaryKPiValid} shows the validity of ($\emp$K) whereas the other axioms are easily dealt with, using the definition of the semantics or the results for $\epsilon$-validities established in \cite{AA} together with Lemma \ref{lemmaRelationPiValidities}. As for the rules, the soundness of (Nec!) results from the definition of always-validities while that of (MP) and (NecK) is standardly established. Finally, Proposition \ref{propLinkValiditiesforAxiomatisation} shows that (R$^\ast$) preserves $\ast$-validity.
\end{proof}

\begin{theorem}[Completeness]\label{Completeness}
$\AAstar$ is complete w.r.t. AA$^\ast$ \emph{i.e.} if $\satisfies ^\ast \phi$ then  $\phi \in \textbf{AA}^\ast$.
\end{theorem}

The proof is provided in the following section.

\section{Completeness of $\AAstar$} \label{sectionCompleteness}

In this section, we show that $\AAstar$ is complete w.r.t. the set of always-validities AA$^\ast$. To do so, we prove the contrapositive of Theorem \ref{Completeness}: from a formula $\phi \notin \AAstar$, we identify a word $\alpha$ and construct a model wherein $[\alpha]\phi$ is not true, thereby demonstrating that $\phi$ is not $\ast$-valid. The proof is based on the method of canonical model construction.

In order to demonstrate Theorem \ref{Completeness}, we first define \emph{theories} and prove some of their important properties.


\begin{definition}[Theory]\label{DefTheory}
    A \emph{theory T} is a set of formulas that satisfies the following conditions:\\
    \indent $(i)$ T contains all the formulas derivable in $AA^\ast$, \textit{i.e.} $AA^\ast \subseteq T$\\
    \indent $(ii)$ T is closed under modus ponens, \textit{i.e.} if  $\phi \in T$ and $\phi \imp \psi \in T$ then $\psi \in T$.
\end{definition}

\begin{definition}[Maximal consistent theory]
    A theory $T$ is \emph{consistent} if and only if $\F \notin T$. A consistent theory $T$ is \emph{maximal consistent} if and only if no consistent theory $T'$ strictly contains $T$.
\end{definition}

Note that the only inconsistent theory is the set $\lang$ of all formulas. Hence, whenever there is a formula $\phi$ such that $\phi \notin T$, the theory $T$ is consistent.

\begin{lemma}\label{lemmaTheories}
    Let $T$ be a theory, $\chi \in \lang$ a formula and $a$ an agent. The following sets are also theories:
    \begin{align*}
        (1)\quad &T + \chi = \lbrace \phi \in \lang \ | \ \chi \imp \phi \in T \rbrace \\
        (2)\quad &K_aT = \lbrace \phi \in \lang \ | \ K_a \phi \in T \rbrace \\
        (3)\quad &[\alpha]T = \lbrace \phi \in \lang \ | \ [\alpha]\phi \in T \rbrace
    \end{align*}
\end{lemma}

\begin{proof}
    We check items $(i)$ and $(ii)$ of Definition \ref{DefTheory}.\\
    $(i)$ We show that all sets contain \textbf{AA$^\ast$}. For this, let $\phi \in \textbf{AA}^\ast$.
    \begin{itemize}
        \item Since $\phi \imp (\chi \imp \phi)$ is a tautology, $\phi \imp (\chi \imp \phi) \in \AAstar$ so $\chi \imp \phi \in \AAstar$ by modus ponens. Then $\chi \imp \phi \in T$, since $T$ is a theory. Therefore $\phi \in T + \chi$. Hence $\textbf{AA}^\ast \subseteq T + \chi$.
        \item By necessitation, $K_a\phi \in \textbf{AA}^\ast$ so $K_a\phi \in T$, because $T$ is a theory. Hence $\phi \in K_a T$. Therefore $\textbf{AA}^\ast \subseteq K_aT $.
        \item Likewise, by necessitation, $[\alpha]\phi \in \textbf{AA}^\ast \subseteq T$ so $\phi \in [\alpha]T$. Hence $\textbf{AA}^\ast \subseteq [\alpha]T$.
    \end{itemize}
    
    $(ii)$ We show that all sets are closed under modus ponens.
    \begin{itemize}
        \item Suppose $\phi, \phi \imp \psi \in T + \chi$. Then by definition $\chi \imp \phi \in T$ and $\chi \imp (\phi \imp \psi) \in T$. Since $(\chi \imp (\phi \imp \psi)) \imp ((\chi \imp \phi) \imp (\chi \imp \psi))$ is a tautology, $(\chi \imp (\phi \imp \psi)) \imp ((\chi \imp \phi) \imp (\chi \imp \psi)) \in T$. By modus ponens, then, $(\chi \imp \phi) \imp (\chi \imp \psi) \in T$ so $\chi \imp \psi \in T$. Hence $\psi \in T + \chi$.

        \item Suppose $\phi, \phi \imp \psi \in K_a T$. By necessitation, $K_a\phi \in T$ and $K_a(\phi \imp \psi)\in T$. Now, by distributivity and modus ponens $K_a\phi \imp K_a\psi \in T$ so $K_a\psi \in T$. Hence $\psi \in K_aT$.

        \item Suppose $\phi, \phi \imp \psi \in [\alpha]T$. By definition, $[\alpha]\phi \in T$ and $[\alpha](\phi\imp\psi) \in T$ so by distributivity and modus ponens $[\alpha]\phi \imp [\alpha]\psi \in T$ so $[\alpha]\psi \in T$. Hence $\psi \in [\alpha]T$.
    \end{itemize}
    \end{proof}

    \begin{lemma}\label{lemmaConsistencyTheories}
        Let $T$ be a theory and $\chi$ a formula. Then $T \subseteq T + \chi$ and $\chi \in T +\chi$. Moreover, if $\lnot \chi \notin T$ then $T + \chi$ is consistent, and if $\chi \notin T$ then $T + \lnot\chi$ is consistent.
    \end{lemma}

\begin{proof}
    If $\phi \in T$ then $\chi \imp \phi \in T$ so $\phi \in T +\chi$. Hence $T \subseteq T + \chi$. Moreover, since $\chi \imp \chi$ is a tautology, $\chi \imp \chi \in T$ so $\chi \in T +\chi$.

    If $\lnot\chi \notin T$ then $\chi \imp \F \notin T$ so $\F \notin T + \chi$. Hence $T + \chi$ is consistent. And if $\chi \notin T$ then $\lnot\lnot \chi \notin T$ (because $\chi \eq \lnot\lnot \chi \in T$ since it is an instance of a tautology) so $\lnot \chi \imp \F \notin T$. Hence $\F \notin T +\lnot\chi$ so $T + \lnot \chi$ is consistent.
\end{proof}

Finally, we can show that

\begin{lemma}\label{lemmaPhiOrNotPhiMCS}
If $\Gamma$ is a maximal consistent theory, then for all formulas $\phi$, either $\phi \in \Gamma$ or $\lnot\phi \in \Gamma$.
\end{lemma}

\begin{proof}
Let $\Gamma$ be a maximal consistent theory and $\phi \in \lang$ a formula such that $\phi \notin \Gamma$ and $\lnot\phi \notin\Gamma$. Then $\Gamma' := \Gamma + \phi$ is consistent and $\Gamma \subsetneq \Gamma'$, which contradicts the fact that $\Gamma$ is maximal consistent. So either $\phi \in \Gamma$ or $\lnot \phi \in \Gamma$.
\end{proof}

\begin{corollary}\label{corollMaxConsTheories}
If $\Gamma$ is a maximal consistent theory, then for all formulas $\phi, \psi$, if $\phi \lor \psi \in \Gamma$ then either $\phi \in \Gamma$ or $\psi \in \Gamma$.
\end{corollary}

\begin{lemma}[Lindenbaum's Lemma]\label{LindenbaumLemma}
If $T$ is a consistent theory, then there is a maximal consistent theory $\Sigma$ such that $T \subseteq \Sigma$. 
\end{lemma}


\begin{proof}
Let $T$ be a consistent theory. Let $\lbrace \phi_k \ | \ k \in \Nat \rbrace$ be an enumeration of the formulas in $\lang$. For each $k \in \Nat$ we construct a consistent theory $T_k$ as follows:
\begin{align*}
&T_0 := T \\
&T_{k+1} := \left\{
	\begin{array}{ll}
        T_k + \phi_k & \mbox{if} \ \lnot\phi_k \notin T_k \\
        T_k & \mbox{otherwise}
    \end{array}
    \right.
\end{align*}
Note that by construction, we get from Lemma \ref{lemmaConsistencyTheories}, for all $k \in \Nat, \ T_k \subseteq T_{k+1}$ and, since $T$ is consistent, each $T_k$ is consistent.

Now we define $\Sigma := \bigcup_{k \in \Nat}T_k$. By construction, $T \subseteq \Sigma$. We show that $\Sigma$ is a maximal consistent theory:
\begin{itemize}
\item $\Sigma$ is a theory: $(i)$ $\AAstar \subseteq T \subseteq \Sigma$ and for $(ii)$, suppose $\phi \in \Sigma$ and $\phi \imp \psi \in \Sigma$. Then there is $k,l \in \Nat$  such that $\phi \in T_k$ and $\phi\imp \psi \in T_l$. Let $n := max(k,l)$. Then $T_k, T_l \subseteq T_n$ and $\phi,\phi\imp\psi \in T_n$. Since $T_n$ is a theory, it is closed by modus ponens so $\psi \in T_n \subseteq \Sigma$. Hence $\psi \in \Sigma$.
\item $\Sigma$ is consistent because each $T_k$ is consistent.
\item $\Sigma$ is maximal consistent: suppose there is a theory $\Sigma'$ such that $\Sigma \subsetneq \Sigma'$. Then there is $k\in \Nat$ such that $\phi_k \in \Sigma'$ but $\phi_k \notin \Sigma$. By construction, that implies $\lnot\phi_k \in T_k$ so $\lnot\phi_k \in \Sigma$  and then $\lnot\phi_k \in \Sigma'$. Hence $\Sigma'$ is not consistent. 
\end{itemize} 
\end{proof}

In order to define a canonical model, we define the following relations.

\begin{definition}\label{defCanonicalRelations}
Let $\chi$ be a formula and $a$ an agent. For $\Gamma,\Delta$ maximal consistent theories, we define the following relations:
\begin{align*}
&\Gamma \equiv_a \Delta & &\text{iff} & &K_a\Gamma \subseteq \Delta  \quad \textit{ i.e. for all formulas } \phi, \quad K_a \phi \in \Gamma \Imp \phi \in \Delta \\
&\Gamma \leq_a \Delta   & &\text{iff} & &[a] \Gamma \subseteq \Delta   \quad \textit{ i.e. for all formulas } \phi, \quad  [a]\phi \in \Gamma \Imp \phi \in \Delta \\
&\Gamma \leq_\chi \Delta   & &\text{iff} & &[\chi] \Gamma \subseteq \Delta   \quad \textit{ i.e. for all formulas } \phi, \quad  [\chi]\phi \in \Gamma \Imp \phi \in \Delta
\end{align*}
If $\alpha = \alpha_1\cdots\alpha_n$, for $\alpha_i$ either an agent or a formula, let $\leq_\alpha \ := \ \leq_{\alpha_1} \circ \cdots \circ \leq_{\alpha_n}$. 
We further define $\equiv$ as the reflexive and transitive closure of the union of all $\equiv_a$.
\end{definition}

We need to stress that $\equiv$ is not an equivalence relation because it is not symmetric. The proof will be given later, since it needs the Existence Lemma stated below.

\begin{lemma}\label{lemmaCompositionWordRelation}
    Let $\Gamma,\Delta$ be maximal consistent theories, and $\alpha,\beta$ be words. Then $[\alpha\beta]\Gamma \subseteq \Delta$, if and only if there is a maximal consistent theory $\Lambda$ such that $[\alpha]\Gamma \subseteq \Lambda$ and $[\beta]\Lambda \subseteq \Delta$.
\end{lemma}

\begin{proof}

Suppose first that $[\alpha\beta]\Gamma \subseteq \Delta$. We need to show that there is a maximal consistent theory $\Lambda$ such that $[\alpha]\Gamma \subseteq \Lambda$ and $[\beta]\Lambda \subseteq \Delta$.

Let $\mathcal{S}= \lbrace \Lambda \text{ theory}\ | \ [\alpha]\Gamma \subseteq \Lambda \text{ and } [\beta]\Lambda \subseteq \Delta \rbrace$.
Obviously, $[\alpha]\Gamma \in \mathcal{S}$.

Note that if $(\Lambda_i)_{i \in I}$ is a chain of elements in $\mathcal{S}$, then $\cup_{i \in I}\Lambda_i \in \mathcal{S}$. Therefore, by Zorn's Lemma, $\mathcal{S}$ has a maximal\footnote{Maximal with respect to inclusion.} element $\Lambda$. Since $\Lambda \in \mathcal{S}$, $\Lambda$ is a theory that contains $[\alpha]\Gamma$ and such that $[\beta]\Lambda \subseteq \Delta$. All that remains to show is that $\Lambda$ is maximal consistent.

We first show that $\Lambda$ is consistent. Suppose $\F \in \Lambda$. Then, $[\beta]\F \in \Lambda$ and, since $[\beta]\Lambda \subseteq \Delta$, $\F \in \Delta$. This contradicts the consistency of $\Delta$. Therefore $\F \notin \Lambda$.

Suppose now that $\Lambda$ is not maximal consistent. Then, there is a formula $\phi$ such that $\phi \notin \Lambda$ and $\lnot \phi \notin \Lambda$. Hence $\Lambda \subsetneq \Lambda + \phi$ and $\Lambda \subsetneq \Lambda + \lnot \phi$. Then $\Lambda + \phi \notin \mathcal{S}$, because $\Lambda$ is a maximal element of $\mathcal{S}$; likewise $\Lambda + \lnot \phi \notin \mathcal{S}$. Since both $\Lambda + \phi$ and $\Lambda + \lnot \phi$ are theories (Lemma \ref{lemmaTheories}) that, obviously, contain $[\alpha]\Gamma$, then $[\beta](\Lambda + \phi) \nsubseteq \Delta$ and $[\beta](\Lambda + \lnot \phi) \nsubseteq \Delta$. Hence, there are formlas $\psi$ and $\chi$ such that
\begin{align*}
	(1)\quad &\psi \in [\beta](\Lambda + \phi) & &\text{ and } & &\chi \in [\beta](\Lambda + \lnot \phi)\\
	(2)\quad &\psi \notin \Delta 			   & &\text{ and } & &\chi \notin \Delta.
\end{align*}
From $(1)$ we conclude that $\phi \imp [\beta]\psi \in \Lambda$ and $\lnot \phi \imp [\beta]\chi \in \Lambda$. Then $[\beta]\psi \lor [\beta]\chi \in \Lambda$ so $[\beta](\psi \lor \chi) \in \Lambda$. Hence $\phi \lor \chi \in \Delta$. Since $\Delta$ is maximal consistent, therefore either $\psi \in \Delta$ or $\chi \in \Delta$ (Corollary \ref{corollMaxConsTheories}). In both cases, we get a contradiction with $(2)$ above. Therefore, $\Lambda$ is maximal consistent. \\

The converse is straightforward.

\end{proof}

\begin{proposition}\label{propCompositionRelation}
Let $\Gamma,\Delta$ be maximal consistent theories, and $\alpha$ a word. Then, $\Gamma \leq_\alpha \Delta$ if and only if $[\alpha]\Gamma \subseteq \Delta$.
\end{proposition}

By convention, $\leq_\epsilon$ is the identity relation.

\begin{proof}
By induction on $\Vert \alpha \Vert$.
\begin{itemize}
\item Cases $\alpha=a$ and $\alpha=\phi$ are straightforward.
\item Case $\alpha =\alpha' a$. We have the following equivalences:
\begin{align*}
\Gamma \leq_\alpha \Delta
&\Eq \Gamma \leq_{\alpha'} \Delta' \text{ and } \Delta' \leq_a \Delta \text{ for some } \Delta'\\
&\Eq [\alpha']\Gamma \subseteq \Delta' \text{ and } \Delta' \leq_a \Delta \text{ for some } \Delta' & &\text{by (IH)}\\
&\Eq [\alpha']\Gamma \subseteq \Delta' \text{ and } [a]\Delta' \subseteq \Delta \text{ for some } \Delta' & &\text{by Definition \ref{defCanonicalRelations}} \\
&\Eq [\alpha'a]\Gamma \subseteq \Delta & &\text{by Lemma \ref{lemmaCompositionWordRelation}}\\
&\Eq [\alpha]\Gamma \subseteq \Delta &
\end{align*}

\item Case $\alpha =\alpha' \chi$. We have the following equivalences:
\begin{align*}
\Gamma \leq_\alpha \Delta
&\Eq \Gamma \leq_{\alpha'} \Delta' \text{ and } \Delta' \leq_\chi \Delta \text{ for some } \Delta'\\
&\Eq [\alpha']\Gamma \subseteq \Delta' \text{ and } \Delta' \leq_\chi \Delta \text{ for some } \Delta' & &\text{by (IH)}\\
&\Eq [\alpha']\Gamma \subseteq \Delta' \text{ and } [\chi]\Delta' \subseteq \Delta \text{ for some } \Delta' & &\text{by Definition \ref{defCanonicalRelations}} \\
&\Eq [\alpha'\chi]\Gamma \subseteq \Delta & &\text{by Lemma \ref{lemmaCompositionWordRelation}}\\
&\Eq [\alpha]\Gamma \subseteq \Delta &
\end{align*}
\end{itemize}
\end{proof}

Note that if $[a]\F \in \Gamma$ (resp. $[\chi]\F \in \Gamma$) then for any maximal consistent theory $\Delta$, $\Gamma \nleq_a \Delta$ (resp. $\Gamma \nleq_\chi \Delta$). Hence, fo all words $\alpha$, if $[\alpha]\F \in \Gamma$ then $\Gamma \nleq_\alpha \Delta$ for any maximal consistent theory $\Delta$.

\begin{lemma}[Existence Lemma]\label{ExistenceLemma}
Let $\Gamma$ be a maximal consistent theory, $\phi,\chi \in \lang$ be formulas and $a \in \A$ an agent.
\begin{align*}
(i) \quad &\text{If } \Ka \phi \in \Gamma \text{ then there is a maximal consistent theory } \Delta \text{ s.t. } \Gamma \equiv_a \Delta \text{ and } \phi \in\Delta.\\
(ii) \quad &\text{If } \langle a \rangle \phi \in \Gamma \text{ then there is a maximal consistent theory } \Delta \text{ s.t. } \Gamma \leq_a \Delta \text{ and } \phi \in\Delta.\\
(iii) \quad &\text{If } \langle \chi \rangle \phi \in \Gamma  \text{ then there is a maximal consistent theory }\Delta \text{ s.t. } \Gamma \leq_\chi \! \Delta \text{ and } \phi \in\Delta.
\end{align*}
\end{lemma}

\begin{proof} \hfill
    \begin{itemize}
        \item[$(i)$] Suppose $\Ka \phi \in \Gamma$. Then $\lnot K_a \lnot \phi \in \Gamma$ so $K_a\lnot \phi \notin \Gamma$. Hence $\lnot\phi \notin K_a\Gamma$.  So, by Lemma \ref{lemmaConsistencyTheories}, $K_a\Gamma + \phi$ is consistent. Now, by Lindenbaum's Lemma, we can extend $K_aT +\phi$ to a maximal consistent theory $\Delta$ such that $K_a\Gamma + \phi \subseteq \Delta$. By Lemma \ref{lemmaConsistencyTheories} again, $K_a\Gamma \subseteq K_a\Gamma + \phi$ so $K_a\Gamma \subseteq \Delta$. Hence $\Gamma \equiv_a \Delta$. Moreover  $\phi \in K_a\Gamma + \phi$ so $\phi \in \Delta$.
        
        \item[$(ii)$] Suppose $\langle a \rangle \phi \in \Gamma$. Then $\langle a \rangle \T \in \Gamma$ so $[a]\F \notin \Gamma$. Hence $\F \notin [a]\Gamma$: $[a]\Gamma$ is consistent. By Lindenbaum's Lemma, we can extend $[a]\Gamma$ to a maximal consistent theory $\Delta$. Since $[a]\Gamma \subseteq \Delta$, $\Gamma \leq_a \Delta$. Moreover, since $\langle a \rangle \phi \in \Gamma$ and, by functionality, also $\langle a \rangle \phi \imp [a]\phi \in \Gamma$, by modus ponens $[a]\phi \in \Gamma$, and then $\phi \in [a]\Gamma$. Hence $\phi \in \Delta$.
        
        \item[$(iii)$] Supose $\langle \chi \rangle \phi \in \Gamma$. Then $\langle \chi \rangle \T \in \Gamma$ so $[\chi]\F \notin\Gamma$. Hence $\F \notin [\chi]\Gamma$ so $[\chi]\Gamma$ is consistent. By Lindenbaum's Lemma, there is a maximal consistent theory $\Delta$ such that $[\chi]\Gamma \subseteq \Delta$. Hence $\Gamma \leq_\chi \Delta$. Moreover, by functionality and modus ponens, $[\chi]\phi \in \Gamma$ so $\phi \in [\chi]\Gamma$ and so $\phi \in \Delta$.
    \end{itemize}
\end{proof}

\begin{corollary}\label{corolExistenceLemmaIFF}
Let $\Gamma$ be a maximal consistent theory and $\alpha$ a word. Then $\langle \alpha \rangle \phi \in \Gamma$ if and only if there is a maximal consistent theory $\Delta$ such that $\Gamma \leq_\alpha \Delta$ and $\phi \in \Delta$.
\end{corollary}

\begin{proof}
Let $\Gamma$ be a maximal consistent theory.
\begin{itemize}
\item[($\Imp$)] The proof of the left-to-right direction proceeds by induction on $\Vert \alpha \Vert$.  Let $\alpha$ be a word such that for all words $\alpha'$, if $\Vert \alpha' \Vert < \Vert \alpha \Vert$ then for all formulas $\phi \in \lang$, if $\langle\alpha'\rangle \phi \in \Gamma$ then there is a maximal consistent theory $\Delta$ such that $\Gamma \leq_{\alpha'} \Delta$ and $\phi \in \Delta$. We show that this holds also for $\alpha$. We distinguish the following cases.
\begin{enumerate}
\item[$\bullet$] Case $\alpha = \epsilon$. Take $\Delta = \Gamma$.
\item[$\bullet$] Case $\alpha = \alpha ' a$. Suppose $\langle \alpha' a \rangle \phi \in \Gamma$, \textit{i.e.} $\langle \alpha' \rangle \langle a \rangle \phi \in \Gamma$. Because $\Vert \alpha' \Vert < \Vert \alpha' a \Vert$, by induction hypothesis, there is a maximal consistent theory $\Delta$ such that $\Gamma \leq_\alpha' \Delta$ and $\langle a \rangle \phi \in \Delta$. Then, by the Existence Lemma, there is a maximal consistent theory $\Delta'$ such that $\Delta \leq_a\Delta'$ and $\phi  \in \Delta'$. Hence $\Gamma \leq_{\alpha'a} \Delta'$ and $\phi \in \Delta'$.
\item[$\bullet$] Case $\alpha = \alpha' \psi$. Suppose $\langle \alpha' \psi \rangle \phi \in \Gamma$, \textit{i.e.} $\langle \alpha' \rangle \langle \psi \rangle \phi \in \Gamma$. Because $\Vert \alpha' \Vert < \Vert \alpha' \psi \Vert$, by induction hypothesis, there is a maximal consistent theory $\Delta$ such that $\Gamma \leq_{\alpha'}\Delta$ and $\langle \psi \rangle \phi \in \Delta$. As above, by the Existence Lemma, we conclude there is a maximal consistent theory $\Delta'$ such that $\Gamma \leq_{\alpha'\psi} \Delta'$ and $\phi \in \Delta'$.
\end{enumerate}

\item[($\Leftarrow$)] For the converse, suppose there is a maximal consistent theory $\Delta$ such that $\Gamma \leq_\alpha \Delta$ and $\phi \in \Delta$. Suppose, towards a contradiction, that $\langle \alpha \rangle \phi \notin\Gamma$. Then $\lnot \langle \alpha \rangle \phi \in \Gamma $ so $[\alpha]\lnot\phi \in \Gamma$. Now, by Proposition \ref{propCompositionRelation}, from $[\alpha]\lnot \phi \in \Gamma$ and $\Gamma \leq_\alpha\Delta$, we get $\lnot\phi \in \Delta$. This, together with $\phi \in \Delta$, contradicts the consistency of $\Delta$. Therefore, $\langle \alpha \rangle \phi \in \Gamma$.
\end{itemize}
\end{proof}

\begin{corollary}\label{corolExistenceLemma2}
Let $\Gamma$ be a maximal consistent theory and $\alpha$ a word. Then $\langle \alpha \rangle \T \in \Gamma$ if and only if there is a maximal consistent theory $\Delta$ such that $\Gamma \leq_\alpha \Delta$.
\end{corollary}

We can now show that $\equiv_a$ is not symmetric, and therefore not an equivalence relation. To show this, we need to find two maximal consistent theories $\Gamma,\Delta$ such that $\Gamma \equiv_a \Delta$ but $\Delta \not\equiv_a\Gamma$. First note that $\lnot [a]\F \land \Ka \T$ is $\ast$-satisfiable. Indeed, take a model $(\lbrace s \rbrace, \sim,V)$ where $p$ is true at $s$, and consider the history only consisting in $p$. Since $p$ is true at $s$, $s\bowtie p$. Obviously $s,p \nvDash [a]\F$ so $s,p \satisfies \lnot [a]\F$, and $s,p \satisfies \Ka \T$ because $s\sim_a s$. Therefore $s,p \satisfies \lnot [a]\F \land \Ka \T$. This shows that $\lnot (\lnot [a]\F \land \Ka \T)$ is not $\ast$-valid. By the Soundness Theorem (Theorem  \ref{Soundness}), then, $\lnot(\lnot [a]\F \land \Ka \T) \notin\AAstar$ so there is a maximal consistent theory $\Gamma$ such that $\lnot [a]\F \land \Ka \T \in \Gamma$. Now, since $\Ka \T \in \Gamma$, by the Corollary \ref{corolExistenceLemma2}, there is a maximal consistent theory $\Delta$ such that $\Gamma \equiv_a \Delta$. We now need to show that $\Delta \not\equiv_a \Gamma$. By Proposition \ref{propDerivabilityKaAllRead}, $K_a[a]\F$ is a theorem, so $K_a[a]\F \in \Delta$. However, $\lnot [a]\F \in \Gamma$ so $[a]\F \notin\Gamma$. Therefore $K_a\Delta \not\subset \Gamma$. Hence $\Delta \not\equiv_a \Gamma$ so $\equiv_a$ is not symmetric.

However, we can restrict the relations $\equiv_a$ to a specific set of maximal consistent theories wherein it would be an equivalence relation. This is what we do in the following definition of a canonical model for asynchronous announcements.
To show completeness of $\AAstar$, it is enough to define a model \emph{associated to} a maximal consistent theory that contains $\emp$.
Such a theory exists because $\lnot\emp \notin \AAstar$: first note that $\lnot\emp$ is not $\ast$-valid because it is not valid (remember that always-validity implies validity, see \ref{propAlwaysValidImpliesValid}). Then, by the soundness of $\AAstar$ (Theorem \ref{Soundness}), $\lnot\emp$ is not derivable: $\lnot\emp \notin\AAstar$. Hence $\AAstar + \emp$ is a consistent theory (by Lemma \ref{lemmaConsistencyTheories}), and we can extend it to a maximal consistent theory, by Lindenbaum's Lemma (\ref{LindenbaumLemma}).

\begin{definition}[Canonical model]\label{defCanonicalModel}
Let $\Sigma$ be a maximal consistent theory such that $\emp \in \Sigma$. The canonical model associated to $\Sigma$ is defined as $M_\Sigma:= (W_\Sigma, (\sim_{\Sigma,a})_{a\in \A}, V_\Sigma)$, where:
\begin{itemize}
\item $W_\Sigma := \lbrace \Gamma \text{ maximal consistent theory} \ | \ \Sigma \equiv \Gamma \rbrace$
\item $\sim_{\Sigma,a}$ is the restriction of $\equiv_a$ to $W_\Sigma$
\item $V_\Sigma(p):= \lbrace \Gamma \in W_\Sigma \ | \ p \in  \Gamma \rbrace$.
\end{itemize}
\end{definition}

First, the following result can be easily proved with axiom ($\emp$K):

\begin{lemma}\label{lemmaPiSigmaToDelta}
Let $\Sigma$ be a maximal consistent theory containing $\emp$, $\Delta \in W_\Sigma$. Then $\emp \in \Delta$.
\end{lemma}

Now, we show that the canonical model is an epistemic model:

\begin{lemma}\label{LemmaRestrictionEquivalenceRelation}
Let $\Sigma$ be a maximal consistent theory such that $\emp \in \Sigma$. Then, for any agent $a\in \A$, $\equiv_a$ restricted to $W_\Sigma$ is an equivalence relation, so all relations $\sim_{\Sigma,a}$ are equivalence relations (for $a\in \A$).
\end{lemma}


\begin{proof}
Let $\Sigma$ be a maximal consistent theory. Suppose $\emp \in \Sigma$. Let $a\in \A$ and $\Gamma, \Delta, \Lambda \in W_\Sigma$. Since $\emp\imp K_a\emp \in \Sigma$, $\emp \in \Gamma$, $\emp \in \Delta$ and $\emp \in \Lambda$. We verify that $\sim_{\Sigma,a}$ is reflexive and Euclidean.
\begin{itemize}
\item We first show that $\sim_{\Sigma,a}$ is reflexive. If $\phi \in K_a\Gamma$ then $K_a\phi \in \Gamma$ so $\phi \in \Gamma$ because $\emp \imp (K_a\phi \imp \phi) \in \Gamma$ and $\emp \in \Gamma$. So $K_a\Gamma \subseteq \Gamma$. Hence $\Gamma \equiv_a \Gamma$.


\item We now show that $\sim_{\Sigma,a}$ is Euclidean. Suppose $\Gamma \equiv_a\Delta$ and $\Gamma \equiv_a \Lambda$. Suppose towards a contradiction that $\Delta \not\equiv_a \Lambda$. Then there is $\phi \in K_a\Delta$ such that $\phi \notin \Lambda$. Hence $\lnot \phi \in \Lambda$. Then, because $K_a \Gamma \subseteq \Lambda$, $K_a \phi \notin \Gamma$, so $\lnot \Ka \lnot \phi \notin \Gamma$.
Hence, $\Ka \lnot \phi \in \Gamma$. Now, from axiom (5), $\Ka \lnot \phi \imp K_a\Ka \lnot \phi \in \Gamma$, so by modus ponens $K_a\Ka\lnot \phi \in \Gamma$. From this and $\Gamma \equiv_a \Delta$, we get $\Ka \lnot \phi \in \Delta$. Hence $\lnot K_a \phi \in \Delta$. But, by hypothesis, $K_a\phi \in \Delta$: this contradicts the consistency of $\Delta$. Therefore, $\Delta \equiv_a \Lambda$.
\end{itemize}
\end{proof}

\begin{corollary}
The relation $\equiv$ is an equivalence relation on $W_\Sigma$.
\end{corollary}

\begin{proof}
By Lemma \ref{LemmaRestrictionEquivalenceRelation} and Definition \ref{defCanonicalModel}, $\sim_{\Sigma,a}$ is an equivalence relation, for all agents $a$. Then, since $\equiv$ is the reflexive and transitive closure of the union of all $\equiv_a$, $\equiv$ restricted to $\Sigma$ is the reflexive and transitive closure of the union of equivalence relations, so it is an equivalence relation itself.
\end{proof}

To show the Truth Lemma (\ref{TruthLemma}), we will use the following lemma.
\begin{lemma}\label{lemmaWordsBoxToDiamond}
Let $\Gamma$ be a maximal consistent theory, $\alpha \in \Words$ a word and $\phi,\psi \in \lang$ be formulas. If $\langle \alpha \rangle \phi \in \Gamma$ and $[\alpha]\psi \in \Gamma$, then $\langle \alpha  \rangle \psi \in \Gamma$.
\end{lemma}

\begin{proof}
Let $\Gamma$ be a maximal consistent theory, $\alpha \in \Words$ a word and $\phi,\psi \in \lang$ be formulas. Suppose $\langle \alpha \rangle \phi \in \Gamma$ and $[\alpha]\psi \in \Gamma$. Since $\langle \alpha \rangle \phi \in \Gamma$, by Corollary \ref{corolExistenceLemmaIFF}, there is a maximal consistent theory $\Delta$ such that $\Gamma \leq_{\alpha} \Delta$ (and $\phi \in \Delta$). Since $[\alpha]\psi \in \Gamma$ then $\psi \in \Delta$. Now, suppose $\langle \alpha \rangle \psi \notin\Gamma$. Then $\lnot\langle \alpha \rangle \psi \in \Gamma$ \emph{i.e.} $[\alpha]\lnot\psi \in \Gamma$ so $\lnot \psi \in \Delta$ which, together with $\psi \in \Delta$, contradicts the consistency of $\Delta$. Therefore, $\langle \alpha \rangle \psi \in \Sigma$.
\end{proof}

\begin{corollary}\label{coWordsBoxtoDiamond}
Let $\Gamma$ be a maximal consistent theory, $\alpha \in \Words$ a word and $\phi \in \lang$ a formula. If $\langle \alpha \rangle \T \in \Gamma$ and $[\alpha]\phi \in \Gamma$ then $\langle \alpha \rangle \phi \in \Gamma$.
\end{corollary}

We can now show the main lemma to prove completeness.

\begin{lemma}[Truth Lemma]\label{TruthLemma}
Let $\Sigma$ be a maximal consistent theory containing $\emp$. Let $\phi$ be a formula and $\alpha$ a word. In the following, we consider $M_\Sigma$, the canonical model associated to $\Sigma$, defined in Definition \ref{defCanonicalModel}.
\begin{itemize}
\item For all $\Lambda \in W_\Sigma$ the following conditions are equivalent:
\begin{enumerate}
\item[(1)] $\langle \alpha \rangle \T \in \Lambda$
\item[(2)] $\Lambda \bowtie \alpha$.
\end{enumerate}
\item For all maximal consistent theories $\Gamma$ such that $\Sigma \ \equiv \circ \leq_\alpha \ \Gamma$, the following conditions are equivalent:
\begin{enumerate}
\item[(i)] $\phi \in \Gamma$
\item[(ii)] For all $\Delta \in W_\Sigma$, if $\Delta \leq_\alpha \Gamma$ then $\Delta,\alpha \satisfies  \phi$
\item[(iii)] There is $\Delta \in W_\Sigma$ such that $\Delta \leq_\alpha \Gamma$ and $\Delta,\alpha \satisfies  \phi$.
\end{enumerate}
\end{itemize} 
\end{lemma}

\begin{proof}

The proof proceeds by $\ll$-induction on $(\alpha,\phi)$. Here, we first show the proof for the first item: $(1) \Eq (2)$. Let $\Lambda \in W_\Sigma$.
\begin{itemize}
\item Case $(\epsilon,\phi)$. By definition, $\T \in \Lambda$ and $\Lambda \bowtie \epsilon$.

\item Case $(\alpha a,\phi)$. Suppose $\langle \alpha a \rangle \T \in \Lambda$, \textit{i.e.} $\langle \alpha \rangle \langle a \rangle \T  \in \Lambda$. Then $\langle \alpha \rangle \T \in \Lambda$. Because $(\alpha,\phi) \ll (\alpha a,\phi)$, by induction hypothesis, $\Lambda \bowtie \alpha$. We now need to show $|\alpha|_a < |\alpha|_!$. Suppose $|\alpha|_a \geq |\alpha|_!$. Then, $\emp \imp [\alpha][a]\F\in \Lambda$ (axiom (Exc!$_3$)). But since $\Lambda \in W_\Sigma$, $\emp \in \Lambda$. So $[\alpha][a]\F \in \Lambda$ by modus ponens. Hence $\lnot\langle \alpha \rangle \langle a\rangle \T \in \Lambda$, which contradicts the fact that $\Lambda$ is consistent. Therefore, $|\alpha|_a < |\alpha|_!$ and since $\Lambda \bowtie \alpha$ we get $\Lambda \bowtie \alpha a$.
Conversely, suppose $\Lambda \bowtie \alpha a$. Then $\Lambda \bowtie \alpha$ and $|\alpha|_! < |\alpha|_a$. Because $(\alpha,\phi) \ll (\alpha a,\phi)$, by induction hypothesis, $\langle \alpha \rangle \T \in \Lambda$. Now, from $|\alpha|_a < |\alpha|_!$ and axiom (Exec$_2$) we get $[\alpha]\langle a \rangle \T \in \Lambda$. Since also $\langle \alpha \rangle \T \in \Lambda$, by Lemma \ref{lemmaWordsBoxToDiamond}, $\langle \alpha \rangle \langle a \rangle \T \in \Lambda$. Hence $\langle \alpha a \rangle \T \in \Lambda$.

\item Case $(\alpha \psi,\phi)$. Suppose $\langle \alpha \psi \rangle \T \in \Lambda$, \textit{i.e.} $\langle \alpha \rangle \langle \psi \rangle \T \in \Lambda$. Then $\langle \alpha \rangle \T \in \Lambda$ so, because $(\alpha,\psi)\ll(\alpha \psi,\phi)$, by induction hypothesis $\Lambda \bowtie \alpha$. Now, since $\langle \alpha \rangle \langle \psi \rangle \T \in \Lambda$, by the Existence Lemma there is $\Lambda'$ a maximal consistent theory such that $\Lambda \leq_\alpha \Lambda'$ and $\langle \psi \rangle \T \in \Lambda'$. Also, from axiom (Exec$_1$), $\psi \eq \langle \psi \rangle \T \in \Lambda'$. Hence, by modus ponens, $\psi \in \Lambda'$. Moreover, since $\Sigma \equiv \Lambda \leq_\alpha \Lambda'$ and $(\alpha,\psi)\ll (\alpha \psi,\phi)$, by induction hypothesis $\Lambda,\alpha \satisfies \psi$. So $\Lambda \bowtie \alpha \psi$.
Conversely, suppose $\Lambda \bowtie \alpha \psi$. So $\Lambda \bowtie \alpha$ and $\Lambda,\alpha \satisfies \psi$. Because  $(\alpha,\psi)\ll (\alpha \psi,\phi)$, by induction hypothesis, $\langle \alpha  \rangle \T \in \Lambda$. So there is a maximal consistent theory $\Lambda'$ such that $\Lambda \leq_\alpha \Lambda'$. Since $\Sigma \equiv \Lambda \leq_\alpha \Lambda'$ and $\Lambda,\alpha \satisfies \psi$, by induction hypothesis again, $\psi \in \Lambda'$. Hence $\langle \alpha \rangle \psi \in \Lambda$. Therefore $\langle \alpha \rangle \langle \psi \rangle \T \in \Lambda$, \textit{i.e.} $\langle \alpha \psi \rangle \T \in \Lambda$.

\end{itemize}

We now show the proof for the second item: $(i)\Eq(ii)\Eq(iii)$. Note that $(ii)$ implies $(iii)$ because $\Sigma \equiv \circ \leq_\alpha \Gamma$ implies that there is $\Delta \in W_\Sigma$ such that $\Sigma \equiv \Delta \leq_\alpha \Gamma$. Hence, we only need to show $(i)\Imp(ii)$ and $(iii)\Imp(i)$.

\begin{itemize}
 
\item Case $(\alpha,p)$.
\begin{enumerate}
    \item[-] $(i)\Imp(ii)$. Suppose $p \in \Gamma$. Let $\Delta \in W_\Sigma$ be such that $\Delta \leq_\alpha \Gamma$ so $\Delta \bowtie \alpha$ by Corollary \ref{corolExistenceLemma2} and the first item of the Truth Lemma. Suppose $\Delta,\alpha \nvDash p$. Then $p \notin \Delta$. Hence $\lnot p \in \Delta$. Since $\lnot p \imp [\alpha]\lnot p \in \Delta$, by modus ponens $[\alpha]\lnot p \in \Delta$. So $\lnot p \in \Gamma$, which together with $p \in \Gamma$ contradicts the consistency of $\Gamma$. Therefore $\Delta,\alpha \satisfies p$.
    \item [-] $(iii) \Imp (i)$. Suppose there is $\Delta \in W_\Sigma$ such that $\Delta \leq_\alpha \Gamma$ and $\Delta,\alpha \satisfies p$. Then $p\in \Delta$. Moreover, $p \imp [\alpha]p \in \Delta$ so by modus ponens $[\alpha]p\in \Delta$. Therefore $p\in\Gamma$.
\end{enumerate}

\item Case $(\alpha,\T)$. Obviously $\T \in \Gamma$ so $(i)$ holds. Suppose $\Delta \in W_\Sigma$ is such that $\Delta \leq_\alpha \Gamma$. Then, by Corollary \ref{corolExistenceLemma2}, $\langle \alpha \rangle \T \in \Delta$ so $\Delta \bowtie \alpha$ by the first item. Hence $\Delta,\alpha \satisfies \T$: $(ii)$ holds and so $(iii)$ does.

\item Case $(\alpha,\lnot\phi)$. To apply the induction hypothesis we use $(\alpha,\phi)\ll (\alpha,\lnot\phi)$.
\begin{enumerate}
    \item[-] $(i)\Imp(ii)$. Suppose $\lnot \phi \in \Gamma$. Let $\Delta \in W_\Sigma$ be such that $\Delta \leq_\alpha \Gamma$. Then $\Delta \bowtie \alpha$. Suppose $\Delta,\alpha \nvDash \lnot\phi$. Then, because $\Delta \bowtie \alpha$ and $\Delta,\alpha \nvDash \lnot \phi$, $\Delta,\alpha \satisfies \phi$. Hence, by induction hypothesis, $\phi \in \Gamma$, which together with $\lnot \phi \in \Gamma$ contradicts the consistency of $\Gamma$. Therefore $\Delta,\alpha \satisfies \lnot \phi$.
    \item [-] $(iii) \Imp (i)$. Suppose there is $\Delta \in W_\Sigma$ such that $\Delta \leq_\alpha \Gamma$ and $\Delta,\alpha \satisfies \lnot \phi$. Then, $\Delta,\alpha \nvDash \phi$ and, by induction hypothesis, $\phi \notin \Gamma$. Hence $\lnot\phi \in \Gamma$.
\end{enumerate}

\item Case $(\alpha,\phi_1 \lor \phi_2)$. We apply the induction hypothesis with $(\alpha,\phi_i)\ll (\alpha,\phi_1 \lor \phi_2)$ for $i \in \lbrace 1,2 \rbrace$.
\begin{enumerate}
    \item[-] $(i)\Imp(ii)$. Suppose $\phi_1 \lor \phi_2 \in \Gamma $. Then $\phi_i \in \Gamma$ for $i=1$ or $i=2$. Let $\Delta \in W_\Sigma$ be such that $\Delta \leq_\alpha \Gamma$. Then $\Delta \bowtie \alpha$. Because, by induction hypothesis, $\Delta,\alpha \satisfies \phi_i$. So $\Delta,\alpha \satisfies \phi_1\lor \phi_2$.
    \item [-] $(iii) \Imp (i)$. Suppose there is $\Delta \in W_\Sigma$ such that $\Delta \leq_\alpha \Gamma$ and $\Delta,\alpha \satisfies \phi_1 \lor \phi_2$. Then $\Delta,\alpha \satisfies \phi_i$ for $i=1$ or $i=2$. Then, by induction hypothesis, $\phi_i \in \Gamma$. Therefore $\phi_1 \lor \phi_2 \in \Gamma$.
\end{enumerate}

\item Case $(\alpha,\Ka \phi)$. To apply the induction hypothesis we use $(\beta,\phi)\ll (\alpha,\Ka\phi)$ and $(\beta, \phi) \ll (\alpha,\Ka \phi)$ for any $\beta$ such that $\alpha \view \beta$. Also note that the dual version of axiom ($\emp$!) is $\emp \imp \Bigl( \langle \alpha \rangle \Ka\phi \eq \Bigl( \langle \alpha \rangle \T \land \bigvee_{\alpha \view \beta} \Ka \langle \beta \rangle \phi  \Bigr) \Bigr)$.

\begin{enumerate}
    \item[-] $(i)\Imp(ii)$. Suppose $\Ka \phi \in \Gamma$. Let $\Delta \in W_\Sigma$ be such that $\Delta \leq_\alpha \Gamma$. Then $\langle \alpha \rangle \T \in \Delta$ and $\Delta \bowtie \alpha$. Moreover, since  $\Delta \leq_\alpha \Gamma$ and $\Ka\phi \in \Gamma$, by Corollary \ref{corolExistenceLemmaIFF}, $\langle \alpha \rangle \Ka \phi \in \Delta$. Now, since $\Delta \in W_\Sigma$, by Lemma \ref{lemmaPiSigmaToDelta}, $\emp \in \Delta$. Moreover, by axiom ($\emp$!), $ \emp \imp \left( \langle \alpha \rangle \Ka \phi \imp \bigvee_{\alpha \view \beta} \Ka \langle \beta \rangle \phi \right) \in \Delta$. So by modus ponens $ \langle \alpha \rangle \Ka \phi \imp \bigvee_{\alpha \view \beta} \Ka \langle \beta \rangle \phi \in \Delta$. Because $\langle \alpha \rangle \Ka \phi \in \Delta$, by modus ponens again, $ \bigvee_{\alpha \view \beta} \Ka \langle \beta \rangle \phi \in \Delta$. Therefore, there is a history $\beta$ such that $\alpha \view \beta$ and $\Ka \langle \beta \rangle \phi \in \Delta$. By the Existence Lemma, there is a maximal consistent theory $\Delta'$ such that $\Delta \equiv_a \Delta'$ and $\langle \beta \rangle \phi \in \Delta'$. Then $\langle \beta \rangle \T \in \Delta'$ and, by induction hypothesis, $\Delta' \bowtie \beta$. Furthermore, by Corollary \ref{corolExistenceLemmaIFF} there is a maximal consistent theory $\Gamma'$ such that $\Delta' \leq_\beta \Gamma'$ and $\phi \in \Gamma'$. Now, $\Sigma \equiv \Delta' \leq_\beta \Gamma'$ so by induction hypothesis $\Delta',\beta \satisfies \phi$. Therefore, from $\Delta \equiv_a \Delta'$, $\alpha \view \beta$, $\Delta' \bowtie \beta$ and $\Delta',\beta \satisfies \phi$, we conclude $\Delta,\alpha \satisfies \Ka \phi$.

    \item [-] $(iii) \Imp (i)$. Suppose there is $\Delta \in W_\Sigma$ such that $\Delta \leq_\alpha \Gamma$ and $\Delta,\alpha \satisfies \Ka \phi$. Then there is a maximal consistent theory $\Delta'$ and a word $\beta$ such that $\Delta \equiv_a \Delta'$, $\alpha \view \beta$, $\Delta' \bowtie \beta$ and $\Delta',\beta \satisfies \phi$. Now, suppose, towards a contradiction, that $\Ka \notin \Gamma$. Then $K_a \lnot \phi \in \Gamma$ so $[\alpha] K_a \lnot \phi \in \Delta$. Moreover, by Lemma \ref{lemmaPiSigmaToDelta}, $\emp \in \Delta$. By axiom ($\emp$!), also $ \emp \imp \left( [\alpha]K_a\lnot \phi \imp [\alpha]\F \lor \bigwedge_{\alpha \view \gamma} K_a [\gamma]\lnot \phi\right) \in \Delta$. By modus ponens, then $[\alpha]\F \lor \bigwedge_{\alpha \view \gamma} K_a[\gamma]\lnot \phi$. Since $\Delta \leq_\alpha \Gamma$, $[\alpha]\T \in \Delta$ so $[\alpha]\F \notin \Delta$. Hence $ \bigwedge_{\alpha \view \gamma}K_a[\gamma]\lnot \phi \in \Delta$. In particular $K_a[\beta]\lnot\phi \in \Delta$. So $[\beta]\lnot\phi \in \Delta'$ because $\Delta \equiv_a \Delta'$. But since $\Delta' \bowtie \beta$, by induction hypothesis, $\langle \beta \rangle \T \in \Delta'$. So by Corollary \ref{corolExistenceLemma2} there is a maximal consistent theory $\Gamma'$ such that $\Delta' \leq_\beta \Gamma'$. Since $[\beta]\lnot\phi \in \Delta'$, $\lnot \phi \in \Gamma'$ so $\phi \notin \Gamma'$. Moreover, $\Sigma \equiv \Delta' \leq_\beta \Gamma'$. Hence, by induction hypothesis, $\Delta',\beta \nvDash \phi$. This contradicts the fact that $\Delta',\beta \satisfies \phi$. Therefore $\Ka \phi \in \Gamma$.
\end{enumerate}

\item Case $(\alpha,\langle a \rangle \phi)$. To apply the induction hypothesis we use $(\alpha a,\phi) \ll (\alpha, \langle a \rangle \phi)$.
\begin{enumerate}
    \item[-] $(i)\Imp(ii)$. Suppose $\langle a \rangle \phi \in \Gamma$. Let $\Delta \in W_\Sigma$ be such that $\Delta \leq_\alpha \Gamma$. Then $\Delta \bowtie \alpha$. Since $\langle a \rangle \phi \in \Gamma$, $\langle a \rangle \T \in \Gamma$ so $[\alpha]\langle a \rangle \T \in \Delta$. Hence $|\alpha|_a < |\alpha|_!$ because otherwise we would have $\emp \imp [\alpha][a]\F \in \Delta$ so $[\alpha][a]\F \in \Delta$ and so $[a]\F \in \Gamma$. Furthermore, since $\langle a \rangle \phi \in \Gamma$ by the Existence Lemma, there is a maximal consistent theory $\Gamma'$ such that $\Gamma \leq_a \Gamma'$ and $\phi \in \Gamma'$. Then $\Sigma \equiv \Delta \leq_{\alpha} \Gamma \leq_a \Gamma'$ so $\Sigma \equiv \Delta \leq_{\alpha a}\Gamma'$. By induction hypothesis, $\Delta,\alpha a\satisfies \phi$. From this and $|\alpha|_a <|\alpha |_!$ we conclude $\Delta,\alpha \satisfies \langle a \rangle \phi$.

    \item [-] $(iii) \Imp (i)$. Suppose there is $\Delta \in W_\Sigma$ such that $\Delta \leq_\alpha \Gamma$ and $\Delta,\alpha \satisfies \langle a \rangle \phi$. Then, $|\alpha|_a < |\alpha|_!$ and $\Delta,\alpha a \satisfies \phi$. Now, since $|\alpha|_a < |\alpha|_!$, by axiom (Exec$_1$), $[\alpha]\langle a \rangle \T \in \Delta$. Hence, since $\Delta \leq_\alpha \Gamma$, $\langle a \rangle \T \in \Gamma$. Then, there is a maximal consistent theory $\Gamma'$ such that $\Gamma \leq_a \Gamma'$, so $\Sigma \equiv \Delta \leq_{\alpha a}\Gamma'$. Because also $\Delta,\alpha a \satisfies \phi$, by induction hypothesis we get $\phi \in \Gamma'$. Hence $[a]\phi \in \Gamma$ and since $\langle a \rangle \T \in \Gamma$, by Lemma \ref{lemmaWordsBoxToDiamond}, $\langle a \rangle \phi \in \Gamma$.
\end{enumerate}

\item Case $(\alpha,\langle \psi \rangle \phi)$. Here we use $(\alpha,\psi)\ll (\alpha, \langle \psi \rangle \phi)$ and $(\alpha \psi,\phi) \ll (\alpha,\langle \psi \rangle \phi)$.
\begin{enumerate}
    \item[-] $(i)\Imp(ii)$. Suppose $\langle \psi \rangle \phi \in \Gamma$. Let $\Delta \in W_\Sigma$ be such that $\Delta \leq_\alpha \Gamma$ so $\Delta \bowtie \alpha$. Since $\langle \psi \rangle \phi \in \Gamma$, $\langle \psi \rangle \T \in \Gamma$ so by axiom (Exec$_1$) $\psi \in \Gamma$. Now, by induction hypothesis, $\Delta,\alpha \satisfies \psi$. Furthermore, since $\langle \psi \rangle \phi \in \Gamma$, by the Existence Lemma, there is a maximal consistent theory $\Gamma'$ such that $\Gamma \leq_\psi \Gamma'$ and $\phi \in \Gamma'$. Then $\Sigma \equiv \Delta \leq_{\alpha \psi} \Gamma'$. By induction hypothesis, $\Delta,\alpha \psi \satisfies \phi$. Because also $\Delta,\alpha \satisfies \psi$, we conclude $\Delta,\alpha \satisfies \langle \psi \rangle \phi$. 
    \item [-] $(iii) \Imp (i)$. Suppose there is $\Delta \in W_\Sigma$ such that $\Delta \leq_\alpha \Gamma$ and $\Delta,\alpha \satisfies \langle \psi \rangle \phi$. Then $\Delta, \alpha \satisfies \psi$ and $\Delta,\alpha \psi \satisfies \phi$. By Proposition \ref{PropExecutabilityHistoryNotFalse} then $\Delta \bowtie \alpha \phi$. By induction hypothesis, $\langle \alpha \psi \rangle \T \in \Delta$ \emph{i.e.} $\langle \alpha \rangle \langle \psi \rangle \T \in \Delta$. Hence $[\alpha]\langle \psi \rangle \T \in \Delta$ so $\langle \psi \rangle \T \in \Gamma$. Then, by the Existence Lemma, there is a maximal consistent theory $\Gamma'$ such that $\Gamma \leq_\psi \Gamma'$. Now, $\Sigma \equiv \Delta \leq_{\alpha \psi} \Gamma'$ with $\Delta,\alpha \psi \satisfies \phi$. So by induction hypothesis $\phi \in \Gamma'$. Therefore $[\psi]\phi \in  \Gamma$ and by Lemma \ref{lemmaWordsBoxToDiamond} $\langle \psi \rangle \phi \in \Gamma$.
\end{enumerate}

\end{itemize}
\end{proof}

We can now prove Theorem \ref{Completeness}.

\begin{proof}[Completeness of \textbf{AA$^\ast$}]
Let $\phi \in \lang$. We need to show that, for all formulas $\phi$, if $\satisfies^\ast \phi$, then $\phi \in \textbf{AA}^\ast$. By contraposition, it is enough to prove that if $\phi \notin$ \textbf{AA$^\ast$} then $\nvDash^\ast \phi$.

Suppose that $\phi \notin$ \textbf{AA$^\ast$}. Then, by rule (R$^\ast$), there is a word $\alpha$ such that $\emp \imp [\alpha]\phi \notin$ \textbf{AA$^\ast$}. So \textbf{AA$^\ast$}$+ \lnot(\emp \imp [\alpha]\phi)$ is consistent. By Lidenbaum's Lemma, we can extend it to a maximal consistent theory $\Sigma$ such that $\lnot(\emp \imp [\alpha]\phi) \in \Sigma$. Then $\emp \in \Sigma$ but $[\alpha]\phi \notin\Sigma$ so $\langle \alpha \rangle \lnot \phi \in \Sigma$. By Corollary \ref{corolExistenceLemmaIFF} there is a maximal consistent theory $\Gamma$ such that $\Sigma \leq_\alpha \Gamma$ and $\lnot \phi \in \Gamma$. We now consider the canonical model $M_\Sigma$ associated to $\Sigma$. From  $\Sigma \equiv \Sigma \leq_\alpha \Gamma$ and $\lnot\phi \in \Gamma$, we conclude by the Truth Lemma that $\Sigma, \alpha \satisfies \lnot \phi$. So $\Sigma\bowtie\alpha$ and $\Sigma,\alpha \nvDash \phi$. Hence $\Sigma,\epsilon \nvDash [\alpha]\phi$. Therefore $\nvDash^\ast\phi$.
\end{proof}

\section{Single-agent case} \label{sectionAxiomOneAgent}

In the preceding sections, we proposed an axiomatisation for always-validities when the set of agents contains at least two distinct agents, by reducing $\ast$-validites to $\epsilon$-validities. Now, we shall explain why this method does not apply in the single-agent case. In the following, we consider the language with only one agent $\lang(\lbrace a \rbrace)$.

Contrary to the multi-agent language (see Lemma \ref{lemmaPiEpsilonWord}), in the single-agent case, there is no formula $\emp'$ that characterises the initial model, \emph{i.e.} the fact that the current history is empty. To demonstrate this, we first show two lemmas.

\begin{lemma}\label{lemmaViewwithTrueA}
Let $\alpha,\beta$ be two words. Then $\alpha \view \beta$ if and only if $\T a \alpha \view \T a \beta$.
\end{lemma}

\begin{proof}
Suppose $\alpha \view \beta$. Then $\alpha \proj_{!a} = \beta\proj_{!a} = \beta\proj_!$ so obviously $(\T a\alpha) \proj_{!a} = (\T a\beta)\proj_{!a} = (\T a\beta)\proj_!$. Hence $\T a \alpha \view \T a \beta$.
Conversely, suppose $\T a\alpha \view \T a\beta$. Then $(\T a\alpha) \proj_{!a} = (\T a\beta)\proj_{!a} = (\T a\beta)\proj_!$ and thus $\alpha \proj_{!a} = \beta\proj_{!a} = \beta\proj_!$. Hence $\alpha \view \beta$.
\end{proof}

\begin{lemma}
For all models $\ModelM$, all states $s\in W$, all words $\alpha$ and all formulas $\phi$, the following equivalences hold:
\begin{align*}
(i) \quad &s\bowtie \alpha	& &\text{if and only if} & &s\bowtie \T a \alpha\\
(ii) \quad &s,\epsilon \satisfies \langle \alpha \rangle \phi & &\text{if and only if} & &s,\T a \satisfies \langle \alpha \rangle \phi.
\end{align*}
\end{lemma}

\begin{proof}
The proof is simultaneously done by $\ll$-induction on $(\alpha,\phi)$. Let $(\alpha,\phi)$ be such that for all $(\alpha',\phi')$, if $(\alpha',\phi')\ll(\alpha,\phi)$ then $(i)$ and $(ii)$ hold. We shall prove that the equivalences also hold for $(\alpha,\phi)$.
We first show $(i)$ by distinguishing the following cases:
\begin{itemize}

\item Case $(\epsilon,\phi)$. Obviously $s\bowtie \epsilon$, $s\bowtie \T a$, $s,\epsilon \satisfies \T$ and $s,\T a\satisfies \T$ so $(i)$ holds:

\item Case $(\alpha a,\phi)$. We have the following equivalences, where the induction hypothesis applies because $(\alpha,\T)\ll (\alpha a,\T)$.
\begin{align*}
s\bowtie \alpha a &\Eq s\bowtie \alpha \text{ and } |\alpha|_a < |\alpha|_! &\\
&\Eq s\bowtie \T a \alpha \text{ and } |\alpha|_a < |\alpha|_! & &\text{ by induction hypothesis}\\
&\Eq s\bowtie \T a \alpha \text{ and } |\T a\alpha|_a < |\T a\alpha|_! &\\
&\Eq s\bowtie \T a \alpha a &
\end{align*}

\item Case $(\alpha \psi,\phi)$. We have the following equivalences, where the induction hypothesis applies because $(\alpha,\psi)\ll(\alpha \psi,\phi)$:
\begin{align*}
s \bowtie \alpha \psi &\Eq s\bowtie \alpha \text{ and } s,\alpha \satisfies\psi &\\
&\Eq s\bowtie \alpha \text{ and } (s\bowtie \alpha \text{ and } s,\alpha \satisfies \psi) &\\
&\Eq s\bowtie \T a \alpha \text{ and } (s\bowtie \alpha \text{ and } s,\alpha \satisfies \psi) & &\text{ by induction hypothesis}\\
&\Eq s\bowtie \T a \alpha \text{ and } s,\epsilon \satisfies \langle \alpha \rangle \psi &\\
&\Eq s\bowtie \T a \alpha \text{ and } s,\T a \satisfies \langle \alpha \rangle \psi  & &\text{ by induction hypothesis}\\
&\Eq s\bowtie \T a \alpha \text{ and } (s\bowtie \T a \alpha \text{ and } s,\T a \alpha \satisfies \psi) &\\
&\Eq s\bowtie \T a\alpha \text{ and } s,\T a\alpha \satisfies \psi &\\
&\Eq s\bowtie \T a \alpha \psi &
\end{align*}
\end{itemize}

We now show $(ii)$.
\begin{itemize}
\item Case $(\alpha,p)$. From $(i)$, obviously, $s,\alpha \satisfies p \Eq s\bowtie \alpha $ and $s \in V(p) \Eq s\bowtie \T a \alpha$ and $ s \in V(p) \Eq s,\T a \satisfies p$.
\item Case $(\alpha, \T)$. Since $s,\epsilon \satisfies \langle \alpha \rangle \T$ if and only if $s\bowtie \alpha$, which is equivalent to $s\bowtie \T a \alpha$ from $(i)$, it is easy to show that this is equivalent to $s,\T a \satisfies \langle \alpha \rangle \T$.


\item Cases $(\alpha, \lnot \phi)$ and $(\alpha,\phi_1 \lor \phi_2)$ are straightforward.



\item Case $(\alpha, \Ka\phi)$. Suppose that $s,\epsilon \satisfies \langle \alpha \rangle \Ka \phi$. That means $s\bowtie \alpha$ and $s,\alpha \satisfies \Ka \phi$ so there is a pair $(t,\beta)$ such that $s\sim_a t$, $\alpha \view \beta$, $t\bowtie\beta$ and $t,\beta \satisfies \phi$. From $s\bowtie \alpha$, $(\alpha,\phi)\ll (\alpha,\Ka\phi)$, we get by induction hypothesis $s\bowtie \T a \alpha$. And, because $(\beta,\phi) \ll (\alpha,\Ka \phi)$, from $t,\beta \satisfies \phi$, we get $t,\T a \beta \satisfies \phi$ by induction hypothesis. Hence, $s\bowtie \T a \alpha$ and $t,\T a \beta \satisfies \phi$ for some $(t,\beta)$ such that $s\sim_a t$, $\alpha \view \beta$ and $t\bowtie \beta$. But $\alpha \view \beta$ is equivalent to $\T a \alpha \view \T a \beta$, by Lemma \ref{lemmaViewwithTrueA}, and, by $(i)$, $t\bowtie \beta$ is equivalent to $t\bowtie \T a \beta$. Hence, we obtain $s,\T a \alpha \satisfies \Ka \phi$. Therefore, $s,\T a \satisfies \langle \alpha \rangle \Ka \phi$. Similarly for the converse.


\item Case $(\alpha,\langle a \rangle \phi)$.  To apply the induction hypothesis, we use $(\alpha a,\phi) \ll (\alpha, \langle a\rangle \phi)$:
\begin{align*}
s,\epsilon \satisfies \langle \alpha \rangle \langle a \rangle \phi 
&\Eq s\bowtie \alpha \text{ and } s,\alpha \satisfies \langle a \rangle \phi\\
&\Eq s\bowtie \alpha \text{ and } s \bowtie \alpha a \text{ and } s,\alpha a\satisfies \phi\\
&\Eq s\bowtie \T a \alpha \text{ and } s \bowtie \T a \alpha a \text{ and } s,\T a \alpha a\satisfies \phi \text{ by (IH)}\\
&\Eq s\bowtie \T a \alpha \text{ and } s,\T a \alpha \satisfies \langle a \rangle \phi \\
&\Eq s,\T a \satisfies \langle \alpha \rangle \langle a \rangle \phi
\end{align*}

\item Case $(\alpha,\langle \phi_1 \rangle \phi_2)$. To apply the induction hypothesis, we use $(\alpha,\phi_1) \ll (\alpha, \langle \phi_1 \rangle \phi_2)$ and $(\alpha\phi_1,\phi_2) \ll (\alpha, \langle \phi_1 \rangle \phi_2)$:
\begin{align*}
s,\epsilon \satisfies \langle \alpha \rangle \langle \phi_1 \rangle \phi_2 
&\Eq s\bowtie \alpha \text{ and } s,\alpha \satisfies \langle \phi_1 \rangle \phi_2\\
&\Eq s\bowtie \alpha \text{ and } s, \alpha \satisfies \phi_1 \text{ and } s,\alpha \phi_1 \satisfies \phi_ 2\\
&\Eq s\bowtie \T a \alpha \text{ and } s, \T a \alpha \satisfies \phi_1 \text{ and } s,\T a \alpha \phi_1\satisfies \phi_ 2 \text{ by IH}\\
&\Eq s\bowtie \T a \alpha \text{ and } s,\T a \alpha \satisfies \langle \phi_1 \rangle \phi_2 \\
&\Eq s,\T a \satisfies \langle \alpha \rangle \langle \phi_1 \rangle \phi_2
\end{align*}
\end{itemize}
\end{proof}

\begin{corollary}
For all formulas $\phi \in \lang(\lbrace a \rbrace)$, $\satisfies \phi \eq \langle \T a \rangle \phi$.
\end{corollary}

We can now prove that in the language with only one agent, there is no formula that expresses the fact that the current history is the empty word $\epsilon$.

\begin{proposition}\label{propNotExistPiSingleAgentCase}
There is no formula $\emp '$ in $\lang(\lbrace a \rbrace)$ such that for all models $\ModelM$, all states $s\in W$ and all words $\alpha \in \Words$, $s,\epsilon \satisfies \langle \alpha \rangle \emp '$ if and only if $\alpha = \epsilon$.
\end{proposition}

\begin{proof}
Suppose that there is a formula $\emp' \in \lang(\lbrace a\rbrace)$ such that for all models $\ModelM$, states $s\in W$ and words $\alpha$, $s,\epsilon \satisfies \langle \alpha \rangle \emp'$ if and only if $\alpha = \epsilon$. Let $\ModelM$ be a model and $s\in W$ a state. By hypothesis $s,\epsilon \satisfies \emp'$. But since $\satisfies \phi \eq \langle \T a \rangle \phi$ for all formulas $\phi$, we get $s,\epsilon \satisfies \langle \T a \rangle \emp'$. But $\T a \neq \epsilon$: this contradicts the hypothesis. Hence, there is no such $\emp'$.
\end{proof}

Such a result relates to Lemma \ref{lemmaPiEpsilonWord} and Proposition \ref{propLinkValiditiesforAxiomatisation} that are crucial to implement the method we used to axiomatise AA$^\ast$. Indeed, $\AAstar$ exploits the possibility to go from any history $\alpha$ back to the empty history through formula $\emp$, precisely because $\langle \alpha \rangle \emp$ is satisfied if and only if $\alpha = \epsilon$. This way, we somehow reduced $\ast$-validites to $\epsilon$-validities. However, Proposition \ref{propNotExistPiSingleAgentCase} states that there is no formula in the language with only one agent that can play the same role as $\emp$ in the multi-agent case. Therefore, in the single-agent case, it is not possible to express the fact that the current history is empty. This is why we cannot directly adapt our axiomatisation to the single-agent case. Some other method would need to be employed here and the axiomatisation is left for future work.
We expect that it is sufficient to replace $\emp$ by $[a]\F$ in that case.

\section{Comparison to other works} \label{sectionComparison}


In this section, we relate our semantics to three-valued logics. We also discuss the notion of cut in distributed computing and relate it to that of history and the notion of cut defined in \cite{knight2019reasoning}. Finally, we discuss choices of epistemic modalities.

\paragraph*{Three-valued logic?}

Although $s,\alpha \satisfies \lnot \phi$ implies $s,\alpha \nvDash \phi$,  the converse need not hold: whereas $s,\alpha \nvDash \F$ always holds, $s,\alpha \satisfies \lnot\F$ only if $s\bowtie \alpha$. Indeed, for a state $s$ wherein $p$ is false, we have $s,\alpha \nvDash \F$ and $s,\alpha \nvDash \lnot \F$  because $s\not\bowtie \alpha$. Our semantics therefore has a flavour of a three-valued logic, with values true, false and undefined, as in Kleene three-valued logics \cite{kleene1952introduction,bochvar1981three,bonzio2025modal}. This appears more clearly if we consider $\bowtie$ as a definability relation between pairs $(s,\alpha)$ and formulas $\phi$ such that $(s,\alpha) \bowtie \phi$ if and only if $s\bowtie \alpha$. However, it is a strange kind of three-valued logic. This is why, curiously, our logic is both weak and strong Kleene three-valued logic. Let us explain. Let the three values be $\mathsf{t}$ (true), $\mathsf{f}$ (false) and $\mathsf{u}$ (unknown). Given that $(s, \alpha) \bowtie \phi$ for all $\phi$ if $s \bowtie \alpha$, any formula is either undefined (has value $\mathsf{u}$) or is defined (has value $\mathsf{t}$ or $\mathsf{f}$). We therefore cannot have that $\phi$ is defined and $\psi$ is undefined so that we would have to choose whether their conjunction, or disjunction, or implication, is undefined (in weak Kleene) or defined (in strong Kleene). The semantics is too rough to distinguish weak from strong. Similarly, $K_a \phi$ is defined if and only if $\hat{K}_a \phi$ is defined. The issue of definability in our semantics rather pertains to asynchronous states $(s,\alpha)$ than to formulas, as in Kleene logics. Therefore, our logic still has validities, contrary to Kleene logics. We do not know if there is a three-valued modal logic that has exactly the features of ours.

\paragraph*{Consistent and inconsistent cuts.}

The notion of history is related to that of a \emph{cut} in a distributed system, where agents are processors communicating with each other. A cut represents the global state of a distributed computation at a particular instant, like a snapshot of a running process. It specifies sequences of events among sending or receiving of messages. For each agent we can distinguish the events before and up to the cut from those that come after the cut \cite{lamport1978ordering}. A cut is \emph{inconsistent} if a message has been received before the cut but was only sent after the cut, and otherwise it is consistent  \cite{panangaden1992concurrent,kshemkalyani2011distributed}. 
In our setting, consistent cuts induce histories: a word $\alpha$ is a history if agents are able to receive only messages that have already been sent. But our agreement/executability relation $\bowtie$ provides a more fine-grained notion of consistency than that of cut. For instance, histories $\alpha := p. \lnot K_a p. a. a$ and $\beta:= p. a. \lnot K_a p. a$ result of  consistent cuts, but $\alpha$ is executable, namely in states wherein $p$ is true, whereas $\beta$ cannot be executed: if $s\bowtie . pa$ then $s,p. a\satisfies K_ap$ so $s,p. a \nvDash \lnot K_a p$ and thus $s\not\bowtie p. a. \lnot K_ap$.

In \cite{knight2019reasoning}, asynchronous communication in a distributed system is modeled through sequences of announcements and notions also called cuts that more directly correspond to the cuts in distributed computing. A cut specifies the number of announcements that each agent has received so far: a state is a triple $(s,\sigma,c)$ where $s$ is a point in a Kripke model, $\sigma$ is a sequence of announcements, and \emph{cut} $c$ lists for each agent the announcements in $\sigma$ that she has received so far. When interpreting epistemic formulas, not all triples $(s,\sigma,c)$ are considered but only those corresponding to {\em consistent states}: consistency is a relation between states $s$ and pairs $(\sigma,c)$, which is similar to our executability relation. Roughly, a state is consistent if the announcements it contains were true when they were made.
As ours, their notion is more fine-grained than the notion of consistent cut in distributed computing. However, the semantics in \cite{knight2019reasoning} allows formulas to be satisfied in inconsistent states, as in \cite{AA}, but unlike in our semantics.

\paragraph*{Belief or knowledge?}\label{sectionBorK}
In our semantics an agent knows something if it holds for all histories she considers possible on the assumption that she has received all information. The agent does not consider it possible that other messages have been sent that she has not yet received but that other agents may have received. Therefore our notion of asynchronous knowledge is not an {\bf S5} notion of knowledge, in particular it does not satisfy the truth axiom $K_a \phi \imp \phi$. However it satisfies $\emp \imp (K_a\phi \imp \phi)$ (the axiom ($\emp${\bf{T}})). We recall that $\emp$ is defined as $\bigwedge_{a\in \A}[a]\F \land \bigwedge_{a,b\in \A} K_a[b]\F$. When all agents know that they all have received all announcements, knowledge is indeed factual. Without that assumption, our epistemic modality is more like one of consistent belief, which is why in prior publications with a similar semantics such as \cite{AA} the notation $B_a \phi$ was used for that, not $K_a \phi$. But given $\emp \imp (K_a\phi \imp \phi)$ we still believe we proposed an acceptable notion of asynchronous knowledge. In \cite{knight2019reasoning}, a notion of asynchronous knowledge is proposed that satisfies $K_a \phi \imp \phi$ as the accessibility relation not only takes received announcements into account but also possible future announcements---in their setting assuming commonly known protocols that is then a more natural assumption. This has the benefit that knowledge is again standard {\bf S5} knowledge, interpreted with an equivalence relation. However, it comes at a price, both theoretically and conceptually. No axiomatisations are proposed in their work for such asynchronous knowledge. Also, given a semantics where messages are public announcements in the logical language wherein positive information grows (of atoms, knowledge of atoms, etc.) but where ignorance may linger forever unless resolved, {\em an agent can never know that another agent is ignorant}: $K_a\lnot(K_b p \lor K_b \lnot p)$ is unsatisfiable in the \cite{knight2019reasoning} semantics. Obtaining an {\bf S5} notion of knowledge by satisfaction and executability (or agreement, consistency, \dots) relations requires imposing structural restrictions to avoid circularity. The authors of \cite{knight2019reasoning} provide two well-founded solutions to this circularity issue either by imposing conditions on the structure of the model and the model transformations, or by restricting the language.

\section{Conclusion and future research}

We presented a multi-agent epistemic logic of asynchronous announcements with epistemic modalities for asynchronous knowledge, modalities for messages sending that correspond to broadcasting announcements, and modalities for individual reception of sent messages by the agents. Formulas are evaluated with respect to a state of an epistemic model and a prior history of such announcements and receptions. We provided an infinitary axiomatisation {\bf AA}$^\ast$ for the always-validies, the formulas that are true in any model after any history of prior events of sending and receiving messages. 

In future research we wish to axiomatise the single-agent case, to which our method cannot be applied, to add group epistemic modalities such as distributed knowledge, and consider non-public communication.
We also envisage introducing other dynamic modalities to represent arbitrary reception or forgetting information, and investigating variations of our semantics, \emph{e.g.}, where agents receive all unread announcements all at once---as when we check our e-mails---or where messages can be broadcast on different channels that might not be accessible to all agents---which could be an asynchronous adaptation of the logics for data exchange and communication of \cite{baltag2024logics}.
Finally, we wish to propose a semantics wherein agents consider possible that announcements have been sent that they have not received yet, so that knowledge is factual.


\nocite{*}
\printbibliography

\newpage

\section*{Appendix: our validities and those of \cite{AA} are the same}



Let $\bowtie_+$ and $\satisfies_+$ denote the agreement and satisfaction relations defined in Section \ref{sectionSemantics} and $\bowtie_-$ and $\satisfies_-$ those defined in \cite{AA}. In this work, $\bowtie_-$ and $\satisfies_-$ are simultaneously defined for states $s$, \emph{histories} $\alpha$ and formulas $\phi$ as follows:

\begin{tabular}{l}
$
\begin{aligned}
&s \bowtie_- \epsilon       & &\text{always} \\
&s \bowtie_- \alpha a      & &\text{iff} & &s \bowtie_- \alpha \text{ and  } |\alpha|_a < |\alpha|_!\\
&s \bowtie_- \alpha\phi    & &\text{iff} & &s\bowtie_- \alpha \text{ and } s,\alpha\satisfies_-\phi  \\ \\
&s,\alpha \satisfies_- p                         & &\text{iff} & &s \in V(p) \\
&s,\alpha \satisfies_- \T                        & &\text{always}  \\
&s,\alpha \satisfies_- \lnot \phi                & &\text{iff} & &s,\alpha \nvDash_- \phi \\
&s,\alpha \satisfies_- \phi \lor \psi               & &\text{iff} & &s,\alpha \satisfies_- \phi \text{ or } s,\alpha \satisfies_- \psi \\
&s, \alpha \satisfies_- \Ka \phi                     & &\text{iff}  & &t,\beta \satisfies_- \phi \text{ for some } (t, \beta) \in W \times\Hist \text{ s.t. } s\sim_at, \alpha \view \beta, t\bowtie_- \beta \\
&s,\alpha\satisfies_-\langle a\rangle\phi            & &\text{iff} & &|\alpha|_a < |\alpha|_! \text{ and } s,\alpha a\satisfies_-\phi\\
&s,\alpha\satisfies_- \langle\phi\rangle\psi          & &\text{iff} & &s,\alpha\satisfies_- \phi \text{ and } s,\alpha\phi\satisfies_- \psi \\
\end{aligned} $ \\
\end{tabular} \\

Since our semantics is defined not only for histories but for \emph{words} in general, it might appear to be more general than that in \cite{AA}. However, it is is fact more constrained, as Proposition \ref{propExecutabilityIfSatisfaction} shows: the satisfying relation $\satisfies$ is restricted to tuples $((s,\alpha),\phi)$ such that $s\bowtie\alpha$.
Nevertheless, both semantics lead to the very same sets of $\epsilon$- and $\ast$-validities.

\begin{proposition}\label{propEquivalenceTwoSemantics}
For all models $\ModelM$, states $s\in W$, all histories $\alpha\in \Words$ and all formulas $\phi \in \lang$, the following equivalences hold:
\begin{align*}
(1) \qquad &s\bowtie_+ \alpha & &\text{iff} & &\alpha s\bowtie_-\alpha \\
(2) \qquad &s,\epsilon\satisfies_+ \phi & &\text{iff} & &s,\epsilon\satisfies_- \phi\\
(3) \qquad &s,\alpha\satisfies_+ \phi & &\text{iff} & &s\bowtie_+\alpha \text{ and } s,\alpha\satisfies_- \phi
\end{align*}
\end{proposition}

\begin{proof}
The proof proceeds by $\ll$-induction on $(\alpha,\phi)$. Let $\ModelM$ be a model, $s\in W$ be a state. We consider a pair $(\alpha,\phi)$ such that for all pairs $(\alpha',\phi')$, if $(\alpha',\phi')\ll(\alpha,\phi)$, then item (1), (2) and (3) hold for $(\alpha,\phi')$. We need to show they also hold for $(\alpha,\phi)$.

We first show (1) by distinguishing cases on $\alpha$.
\begin{itemize}
\item If $\alpha = \epsilon$, by definition, $s\bowtie_+ \epsilon$ and $s\bowtie_- \epsilon$.
\item If $\alpha = \alpha' a$, we have the following equivalences, where the induction hypothesis applies because $(\alpha',\phi)\ll (\alpha' a,\phi)$.
\begin{align*}
s\bowtie_+ \alpha' a &\Eq s\bowtie_+\alpha \text{ and } |\alpha'|_a < |\alpha'|_!\\
					&\Eq s\bowtie_-\alpha' \text{ and } |\alpha'|_a < |\alpha'|_! \text{ by (IH)} \\
					&\Eq s\bowtie_- \alpha' a
\end{align*}
\item If $\alpha = \alpha'\psi$, we have the following equivalences, where the induction hypothesis applies because $(\alpha',\phi)\ll (\alpha' \psi,\phi)$.
\begin{align*}
s\bowtie_+ \alpha' \psi &\Eq s\bowtie_+\alpha' \text{ and } s,\alpha' \satisfies_+ \psi\\
					   &\Eq s\bowtie_-\alpha' \text{ and } s,\alpha'\satisfies_- \psi \text{ by (IH)} \\
					   &\Eq s\bowtie_- \alpha' \psi
\end{align*}
\end{itemize}
Now, we write $\bowtie$ for $\bowtie_-$ and $\bowtie_+$. We omit the proof of the second item as it is elementary. Concerning the third, we distinguish cases on $\phi$, where obvious inductive cases are omitted.

\begin{itemize}
\item If $\phi = p$, obviously, $s,\alpha\satisfies_+ p \Eq s\bowtie\alpha \text{ and } s\in V(p) \Eq s\bowtie\alpha \text{ and } s,\alpha \satisfies_-p$.

\item If $\phi = \T$, by definition $s,\alpha\satisfies_+\T \Eq s\bowtie\alpha \Eq s\bowtie\alpha \text{ and } s,\alpha\satisfies_-\T$.

%
%

\item If $\phi = \M_a\phi'$, the induction hypothesis applies because $(\beta,\phi')\ll(\alpha,\Ka\phi')$ for all $\beta$ such that $\alpha \view \beta$. We then get:
\begin{align*}
	&s,\alpha \satisfies_+ \M_a\phi' \\
\Eq \ &s\bowtie\alpha \text{ and } t,\beta \satisfies_+ \phi' \text{ for some } t,\beta \text{ s.t. } s \sim_a t, \alpha \view \beta \text{ and } t\bowtie\beta \\
\Eq \ &s\bowtie\alpha \text{ and } t\bowtie\beta \text{ and } t,\beta \satisfies_- \phi' \text{ for some } t,\beta \text{ s.t. } s \sim_a t, \alpha \view \beta \text{ and } t\bowtie\beta \text{ by (IH)}\\
\Eq \ &s\bowtie\alpha \text{ and } t\bowtie\beta \text{ and } t,\beta \satisfies_- \phi' \text{ for some } t,\beta \text{ s.t. } s \sim_a t \text{ and } \alpha \view \beta \\
\Eq \ &s\bowtie\alpha \text{ and } t,\beta \satisfies_- \phi' \text{ for some } t,\beta \text{ s.t. } s \sim_a t \text{ and }\alpha \view \beta \\
\Eq \ &s\bowtie\alpha \text{ and } s,\alpha \satisfies_- \M_a\phi'
\end{align*}

\item If $\phi = \dia{a}\phi'$, the induction hypothesis applies because $(\alpha a,\phi')\ll (\alpha,\langle a \rangle \phi')$. This is straightforward.


\item If $\phi = \dia{\phi'}\phi''$, the induction hypothesis applies because $(\alpha,\phi')\ll(\alpha,\langle \phi' \rangle\phi'')$ and $(\alpha \phi',\phi'')\ll (\alpha,\langle \phi' \rangle \phi'')$. We also use the fact that $s\bowtie\alpha\phi'$ if and only if $s\bowtie\alpha$ and $s,\alpha\satisfies_-\phi'$. We then get the following equivalences:
\begin{align*}
&s,\alpha \satisfies_+ \dia{\phi'}\phi''\\
\Eq\ &s,\alpha \satisfies_+ \phi' \text{ and } s,\alpha\phi' \satisfies_+ \phi'' \\
\Eq\ &(s\bowtie\alpha \text{ and } s,\alpha \satisfies_- \phi') \text{ and } (s\bowtie\alpha\phi' \text{ and } s,\alpha\phi' \satisfies_- \phi'') & &\text{by (IH)}\\
\Eq\ &(s\bowtie\alpha \text{ and } s,\alpha \satisfies_- \phi') \text{ and } ((s\bowtie\alpha \text{ and } s,\alpha \satisfies_- \phi') \text{ and } s,\alpha\phi' \satisfies_- \phi'')\\
\Eq\ &s\bowtie\alpha \text{ and } (s,\alpha \satisfies_- \phi' \text{ and } s,\alpha\phi' \satisfies_- \phi'') \\
\Eq\ &s\bowtie\alpha \text{ and } s,\alpha \satisfies_- \dia{\phi'}\phi''
\end{align*}
\end{itemize}
\end{proof}

\begin{corollary}\label{corollSameValidities}
For all formulas $\phi \in \lang$:
\begin{align*}
(1) \quad &\satisfies_+ \phi \quad \text{ if and only if } \quad \satisfies_-\phi\\
(2) \quad &\satisfies^\ast_+ \phi \quad \text{ if and only if } \quad \satisfies^\ast_-\phi.
\end{align*}
\end{corollary}

\begin{proof}
$(1)$ is directly obtained from Proposition \ref{propEquivalenceTwoSemantics}. Concerning $(2)$, note that, if $\alpha$ is not a history, then $[\alpha]\phi$ is trivially valid, for all formulas $\phi$. Indeed, for any model $\ModelM$ and any state $s\in W$, if $\alpha$ is not a history, $s\not\bowtie \alpha$, by Proposition \ref{lemmaBowtieHistory} and hence $s, \epsilon
\satisfies [\alpha] \phi$, by definition. Therefore, $\satisfies_+^\ast \phi$ if and only if $\satisfies_+^\ast [\alpha]\phi$ for all \emph{histories} $\alpha$. Now we can conclude from $(1)$ that $\satisfies^\ast_+ \phi$ if and only if $\satisfies^\ast_- \phi$.
\end{proof}

This shows that both semantics define the same sets of validities AA and AA$^\ast$.

\end{document}